\documentclass[11pt, reqno]{amsart}

\usepackage[english]{babel}
\usepackage{amsthm} 
\usepackage{amssymb}
\usepackage{amsmath}
\usepackage{hyperref}
\usepackage{amsfonts}
\usepackage{mathtools}
\usepackage{cleveref}
\usepackage{autonum}
\usepackage{stmaryrd}
\usepackage{fancyhdr}
\usepackage{caption}
\usepackage[ruled,vlined]{algorithm2e}

\usepackage{natbib}

\usepackage{graphicx,color}
\usepackage{epstopdf}
\usepackage[left=0.9in, right=0.9in, top=1.1in, bottom=1.1in]{geometry}
\usepackage{mathrsfs}  
\usepackage{subcaption}
\usepackage{layouts}
\usepackage{soul}
\allowdisplaybreaks

\newtheorem{theorem}{Theorem}[section]
\newtheorem{proposition}[theorem]{Proposition}
\crefname{proposition}{proposition}{propositions}
\Crefname{Proposition}{Proposition}{Propositions}

\crefname{corollary}{corollary}{corollaries}
\Crefname{Corollary}{Corollary}{Corollaries}

\newtheorem{example}[theorem]{Example}
\crefname{example}{example}{examples}
\Crefname{Example}{Example}{Examples}

\crefname{figure}{figure}{figures}
\Crefname{Figure}{Figure}{Figures}
\newtheorem{lemma}[theorem]{Lemma}

\theoremstyle{definition}

\newtheorem{definition}[theorem]{Definition}

\newtheorem{assumption}[theorem]{Assumption}

\theoremstyle{remark}
\newtheorem{remark}[theorem]{Remark}

\newcommand{\N}{\mathbb{N}}

\newcommand{\cL}{\mathcal{L}}

\newcommand{\cLtilde}{\widetilde{\cL}}

\newcommand{\gammatilde}{\widetilde{\gamma}}

\newcommand{\mutilde}{\tilde\mu}

\newcommand{\Ttilde}{\tilde T}
\newcommand{\lambdatilde}{\tilde\lambda}
\newcommand{\pot}{\mathrm{U}}
\newcommand{\pottilde}{\tilde\pot}
\newcommand{\Qtilde}{\widetilde Q}
\newcommand{\vftilde}{\tilde\vf}
\newcommand{\varphitilde}{\tilde\varphi}

\newcommand{\Rtilde}{\tilde R}
\newcommand{\btilde}{\tilde b}
\newcommand{\sigmatilde}{\tilde\sigma}
\newcommand{\rmB}{\mathrm{B}}

\newcommand{\rmBP}{\mathrm{BP}}
\newcommand{\rmH}{\mathrm{H}}
\newcommand{\rmJ}{\mathrm{J}}
\newcommand{\rmE}{\mathrm{E}}
\newcommand{\rmM}{\mathrm{M}}

\newcommand{\rmV}{\mathrm{V}}
\newcommand{\rmW}{\mathrm{W}}

\newcommand{\rmF}{\mathrm{F}}

\newcommand{\speed}{s}

\newcommand{\cD}{\mathcal{D}}

\newcommand{\Ybar}{\overline{Y}}

\newcommand{\R}{\mathbb{R}}
\newcommand{\vf}{\Phi}

\newcommand{\Exp}{\mathrm{Exp}}
\newcommand{\PP}{\mathbb{P}}
\newcommand{\PE}{\mathbb{E}}

\newcommand{\sign}{\textnormal{sign}}

\newcommand{\Unif}{\mathrm{Unif}}
\usepackage[dvipsnames]{xcolor}

\usepackage{booktabs} 
\usepackage{array} 
\usepackage{paralist} 
\usepackage{verbatim} 
\usepackage{fancyvrb}
\usepackage{float}
\usepackage{caption}
\usepackage{subcaption}
\usepackage{bbm} 

\usepackage{hyperref}
\usepackage{url}
\usepackage[toc,page]{appendix}
\usepackage{color}
\usepackage{nicefrac}
\usepackage[shortlabels]{enumitem} 

\newcommand{\dd}{\textnormal{d}}
\newcommand{\1}{\mathbbm{1}}

\begin{document}
\title{Sampling with time-changed Markov processes}
\author[Andrea Bertazzi and Giorgos Vasdekis]{}


\date{}

\maketitle
\vspace{-10pt}
\begin{center}
    \textsc{Andrea Bertazzi} \\
    CMAP, Ecole polytechnique,  France\\
    \texttt{andrea.bertazzi@polytechnique.edu} \\
    \bigskip
    \textsc{Giorgos Vasdekis} \\
	School of Mathematics, Statistics and Physics, Newcastle University, U.K. \\
    \texttt{giorgos.vasdekis@newcastle.ac.uk}
\end{center}

\begin{abstract}
We study time-changed Markov processes to speed up the convergence of Markov chain Monte Carlo (MCMC) algorithms. The time-changed process is defined by adjusting the speed of time of a base process via a user-chosen, state-dependent function. We explore the properties of such transformations and apply this idea to several Markov processes from the MCMC literature, such as Langevin diffusions and piecewise deterministic Markov processes, obtaining novel modifications of classical algorithms and also re-discovering known MCMC algorithms. We prove theoretical properties of the time-changed process under suitable conditions on the base process, focusing on connecting the stationary distributions and qualitative convergence properties such as geometric and uniform ergodicity, as well as a functional central limit theorem. We also provide a comparison with the framework of space transformations, clarifying the similarities between the approaches. Throughout the paper we give various visualisations and numerical simulations on simple tasks to gain intuition on the method and its performance.
Finally, we provide numerical simulations to gain intuition on the method and its performance on benchmark problems. Our results indicate a performance improvement in the context of multimodal distributions and rare event simulation.
\end{abstract}

\maketitle

\section{Introduction}

The need to generate samples from complex probability distributions is ubiquitous in Bayesian statistics and beyond (e.g. \cite{brooks2011handbook}, \cite{Stuart}). The standard technique for achieving this is Markov chain Monte Carlo (MCMC), in which a Markovian process is constructed such that its distribution converges to the target distribution as the process evolves. The outcomes of the process are then used as samples from the distribution of interest. While the original MCMC algorithms used discrete-time processes, continuous-time processes have since gained popularity. These include diffusion processes \citep{MALA}, Markov jump processes \citep{green:95,power2019acceleratedsamplingdiscretespaces}, and piecewise deterministic Markov processes (PDMPs) \citep{fearnhead2018}. 
However, many of these processes struggle with targets commonly encountered in practical settings, such as multi-modal or heavy tailed distributions. For example, in the case of multi-modal targets, the processes tend to be attracted to the mode they find themselves in, to the extent that they struggle to hop to another mode. In the case of heavy-tailed distributions, they fail to explore the tails of the distribution in a systematic and consistent way.



In this paper, we introduce and discuss the framework of sampling using {\it time-changed Markov processes} to obtain more efficient continuous-time MCMC algorithms for challenging tasks, such as those described above. The fundamental idea is to apply a {\it random time-change} to a {\it base} Markov process in such a way that time accelerates in regions where we want the process to visit more frequently. At the same time, this ensures that the process does not spend too much time in these regions, which is crucial for the process to satisfy a law of large numbers. The interpretation of these regions depends on the specific task at hand. For instance, they may correspond to low-density areas that separate modes in a multimodal target setting or to the tails of the target distribution in the case of rare events estimation. In practice, the time-change is regulated by a \emph{speed function}, $s$, which maps the state space to strictly positive real numbers and accelerates or decelerates time based on the current state of the process.


Throughout the paper, we assume that the base process, denoted by $Y$, is a well-studied Markov process. We provide conditions on $Y$ and $s$ that ensure desirable properties of the time-changed process $X$, such as the law of large numbers (LLN), qualitative convergence rates to the stationary distribution, and functional central limit theorems. In particular, suitable choices of $s$ ensure that the convergence of $X$ is exponential in time, even if $Y$ converges only polynomially or more slowly. Moreover, we provide conditions for the uniform ergodicity of $X$, which refers to the exponential convergence of its law to the target, independently of the initial condition. Throughout the paper, we highlight the numerous connections to the existing literature and apply our framework to several Markov processes that are widely used in the context of MCMC algorithms.

\subsubsection*{Related work}
The idea of sampling from a measure $\mu$ with a time-change of a Markov process was first discussed by \citet{roberts_stramer_2002}. The authors consider a time-changed overdamped Langevin diffusion corresponding to speed functions of the particular form $s(x)=\mu^{-\alpha}(x)$ for $\alpha\in(0,1)$, where $\mu$ is the target distribution. The focus of \citet{roberts_stramer_2002} is then limited to a specific choice of process and speed function, whereas in this paper we define and study the framework in generality.
The time-change of the Zig-Zag process (ZZP), a particular PDMP, was already considered by \citet{Vasdekis_speedup}, but the process was introduced with ad-hoc arguments and without indicating this connection.
The approach of time-change is also deeply connected to importance sampling \citep{robert_monte_2004,mcbook}. Indeed, computing ergodic averages with respect to a time-changed process is equivalent to obtaining re-weighted ergodic averages of the base process for a random time horizon (see Equation \eqref{eq:lln_equivalence} below). 
In this direction, \cite{chak2023optimal} studies the optimal choice of ``proposal distribution" $\mutilde$ to use as target for the overdamped Langevin diffusion, then adjusting the obtained sample by importance weights. We shall revisit this connection in \Cref{ex:overdamped}. It is worth mentioning that the speed function $s(x)=\mu^{-\alpha}(x)$ with $\alpha = 1$ was shown to be optimal across all diffusion coefficients for the overdamped Langevin diffusion in a specific (multimodal) setting when the state space is included in a compact set \citep{lelièvre2024optimizingdiffusioncoefficientoverdamped}. Time-changed Langevin diffusions were also considered by \citet{Leimkuhler2024}, who established some first theoretical properties of such processes and proposed suitable discretisation schemes for their practical simulation.
In the discrete-time world, \cite{andral2023importance} use as base process a Markov chain with stationary distribution different than the target and adjust it by letting it spend a random number of iterations at each state. This is closely connected to some of the ideas we discuss in \Cref{sec:estimate_expectations}. Similar ideas were introduced by \cite{MALEFAKI20081210} in a semi-Markov implementation.
The frameworks of time changes and space transformations are similar when the main goal is to increase the frequency of visits to the tails, or equivalently to obtain a uniformly ergodic process. This aspect is discussed in \Cref{sec:spacetransformations}, where we particularly focus on the connection to \cite{Johnson_Geyer} and \cite{yang2024stereographicmarkovchainmonte}.
Finally, we note that a first version of this work appeared in the PhD thesis of the first author \citep{bertazzi2023approximations}.

\subsubsection*{Contributions and organisation of the paper}
\begin{itemize}[leftmargin=5mm]
	\item \Cref{sec:general_framework} lays the mathematical foundations of our framework. The main result of this section, \Cref{thm:lln_timechange}, shows that $X$ has $\mu$ as stationary distribution and satisfies a LLN whenever the base process, $Y$, satisfies a LLN with respect to the probability distribution $\mutilde(\dd y) \propto s(y)\mu(\dd y)$. In \Cref{sec:estimate_expectations} we discuss three different approaches to sample and estimate expectations with respect to $\mu$ using a time-changed Markov process. While two of these approaches require simulating a continuous time process, the third approach takes advantage of a jump process that is based on discretisations of $Y$ and is straightforward to implement. 
	\item \Cref{sec:examples} is devoted to the application of our framework to several Markov processes that are widely used in the MCMC literature, such as piecewise deterministic Markov processes (PDMPs) and diffusion processes. We also introduce our running example: the Zig-Zag process (ZZP) \citep{ZZ}. We show that applying a  time-change to the ZZP gives the speed-up ZZP, a variation of the ZZP recently introduced and analysed by \cite{Vasdekis_speedup}. 
	\item \Cref{sec:theory} contains our theoretical results on the convergence properties of time-changed Markov processes under suitable conditions on the base process and on the speed function. In particular, \Cref{thm:ergodicity_markovprocess} obtains cases when $X$ is geometrically or uniformly ergodic under weak assumptions on $Y$. \Cref{thm:uniform_ergodicity_pdmps} gives simpler conditions for uniform ergodicity in the specific setting of PDMPs.
    Moreover, \Cref{thm:FCLT} gives conditions that ensure $X$ satisfies a functional central limit theorem, expressing its asymptotic variance in terms of the asymptotic variance of the base process. We apply our results to our running example in \Cref{thm:ergodicity_ZZP} and \Cref{thm:FCLT_ZZP}. Finally, in \Cref{sec:conv.mjp} we specialise the developed theory to the case where $Y$ is a particular jump process, giving conditions on its jump kernel that imply convergence to the target distribution.
	\item \Cref{sec:spacetransformations} discusses the relation between time changes and space transformations, illustrating the similarities and differences of the two methods. In particular, we shall observe that the speed function in our framework plays a similar role as the Jacobian determinant of the diffeomorphism that defines the space transformations. We analyse this connection considering a novel diffeomorphism, which maps $\R^d$ to points inside the $d$-dimensional unit sphere.
	\item \Cref{sec:numerical_simulations} contains numerical simulations on toy distributions that show the expected benefits of the framework.
    \item \Cref{{sec:discussion}} contains a brief discussion, summarizing the results of this paper and suggesting potential directions for future work.
    \item The appendices contain the proofs of our theoretical results, along with additional numerical simulations.
\end{itemize}

\section{Framework}\label{sec:general_framework}

In this section we lay the foundations for using time-changed Markov processes to sample from a target probability distribution $\mu$. \Cref{sec:TT_setup} introduces a process $X$ as a time-change of a base process $Y$. The main result of this section, \Cref{thm:lln_timechange}, gives weak conditions on $Y$ that ensure $X$ has $\mu$ as stationary distribution, and moreover that it satisfies a strong law of large numbers. \Cref{sec:estimate_expectations} discusses different strategies to take advantage of time-changed Markov processes to estimate expectations of observables with respect to $\mu$. 

\subsection{Time-changes of Markov processes}\label{sec:TT_setup}

We propose to sample from a probability measure $\mu$ on $\rmE$ using a Markov process $X$ defined as a time change of a base process $Y$, in the sense that 
\begin{equation}\label{eq:time_change}
	X_t =Y_{r(t)}, \text{ where } \,\,	r(t) := \int_0^t s(X_u) \dd u.
\end{equation}
Throughout the document, $Y$ will be a càdlàg, time-homogeneous Markov process with state space $\rmE$ and should be interpreted as a well-studied process such as a Langevin diffusion \citep{MALA}, or a piecewise deterministic Markov process (PDMP) \citep{Davis1984} such as the Zig-Zag process (ZZP) \citep{ZZ}. 
The function $s$, called the \emph{speed function}, regulates the time change applied to $Y$. We require that $s$ satisfies the following condition.
\begin{assumption}\label{ass:s.integrability}
    The speed function $s: \rmE \rightarrow \mathbb{R}_{+}$ is a continuous function, it satisfies
    \begin{equation}\label{s.integrability.assumption:00}
        \mu(s):= \int_{\rmE} s(y)\mu(\dd y) < \infty,
    \end{equation}
    and there exists a constant $\underline{s} >0$ such that $s(x) \geq \underline{s}$ for all $x \in \rmE$.
\end{assumption}
\noindent We observe that since $s$ is lower bounded by a constant, a change of variables gives $t =\int_0^{r(t)} \frac{1}{s(Y_u)} \dd u,$ 
and thus we find the inverse relation
\begin{equation}\label{eq:time_transformed_process_inverse}
	Y_t =X_{ r^{-1}(t)}\,\,\, \text{ for } \,\,\, r^{-1}(t) := \int_0^{t} \frac{1}{s(Y_u)} \dd u.
\end{equation}

Intuitively, the time-changed process $X$ follows the same paths of $Y$, with the modification that it speeds up when $s$ is large and it slows down when $s$ is small. As a consequence, the two processes cannot have the same invariant measure. In particular, for $X$ to target $\mu$, we now show that $Y$ should have stationary distribution 
\begin{equation}\label{mu.tilde.def:1}
    \mutilde(\dd x) := \frac{1}{\mu(s)} s(x) \mu(\dd x),
\end{equation}
that is a well-defined probability measure under \Cref{ass:s.integrability}. 
The key assumption to ensure $X$ is $\mu$-stationary is that the base process $Y$ satisfies a law of large numbers (LLN) with respect to $\mutilde$.
\begin{assumption}[LLN for the base process]\label{ass:lln_base}
 For any $f \in L^1(\mutilde)$, and any initial condition $x \in \rmE$,
 \begin{equation}\label{lln.base:1}
     \frac{1}{T}\int_0^T f(Y_u) \dd u \,\, \xrightarrow[a.s.]{T \rightarrow \infty} \,\, \int_{\rmE} f(y) \mutilde(\dd y).
 \end{equation}
\end{assumption}
The assumptions we have introduced above are sufficient to obtain that $X$ has stationary distribution $\mu$ and also that it satisfies a LLN.
\begin{theorem}[Invariance and LLN for the time-changed process]\label{thm:lln_timechange}
    Suppose \Cref{ass:s.integrability} and \Cref{ass:lln_base} hold. Then, the process $X$ has $\mu$ as unique stationary distribution. Furthermore, for any $f \in L^1(\mu)$, and all initial conditions $x \in \rmE$, we have  
\begin{equation}\label{lln.transform:1}
     \frac{1}{T}\int_0^T f(X_t) \dd t \,\, \xrightarrow[a.s.]{T \rightarrow \infty} \,\, \int_{\rmE} f(y)  \mu(\dd y).
 \end{equation}
 \end{theorem}
\begin{proof}
    The proof can be found in \Cref{sec:proof_LLN}.
\end{proof}
The weakness of the required assumptions guarantees that the results of \Cref{thm:lln_timechange} holds essentially in all cases of practical interest. Therefore, the time-changed process $X$ can generally be used to sample from a given target distribution and to estimate its expectations.

\subsection{Estimating expectations with time-changed Markov processes}\label{sec:estimate_expectations}
In this section we describe three asymptotically exact approaches to estimate $\mu(f):= \int_{\rmE} f(y) \mu(\dd y)$. These are based respectively on (1) the direct simulation of $X$, (2) a re-weighting of a path of $Y$, (3) a jump process $X$ with biased jump kernel $\Qtilde$ and rate $s$. 

\begin{enumerate}[wide, labelwidth=!, labelindent=5pt, label=(\arabic*)]
    \item When it is possible directly simulate the time-changed process $X$, one can rely on \Cref{thm:lln_timechange} and use  the ergodic average
    \begin{equation}
        \frac{1}{T} \int_0^T f(X_t) \dd t,
    \end{equation}
    or alternatively any discrete ergodic average obtained with a uniform step size. This is the approach used e.g. by \citet{Vasdekis_speedup,roberts_stramer_2002}. 
    \item In cases where simulating the process $X$ is challenging, we can estimate $\mu(f)$ by re-weighting paths of the base process $Y$. Indeed, following the proof of \Cref{thm:lln_timechange} (see \Cref{sec:proof_LLN}),
    the change of variables $ u = r(t)$, with Jacobian $s(Y_u)$, gives
    \begin{equation}\label{eq:lln_equivalence}
    	\frac{1}{T}\int_0^T f(X_t) \dd t \,=\, \frac{1}{T}\int_0^{r(T)} \frac{f(Y_u)}{s(Y_u)} \dd u.
    \end{equation}
    \Cref{thm:lln_timechange} guarantees that the right hand side of \eqref{eq:lln_equivalence} is a consistent estimator of $\mu(f)$.
The estimator \eqref{eq:lln_equivalence} requires the simulation of $Y$ for a random time $r(T)$, which is not always streightforward to calculate. Alternatively, we may fix a deterministic time horizon $T$ and use the estimator
\begin{equation}\label{eq:erg_avg_ctstime_Y_compbudg}
	\frac1{r^{-1}(T)} \int_0^{T} \frac{f(Y_t)}{s(Y_t)} \dd t \,\,=\,\, \frac{\int_0^{T} \frac{f(Y_t)}{s(Y_t)} \dd t }{\int_0^{T} \frac{1}{s(Y_t)} \dd t }\,\,,
\end{equation}
which is equivalent to simulating $X$ until time $r^{-1}(T)$. For the right hand side of \eqref{eq:erg_avg_ctstime_Y_compbudg} we used the expression given by \eqref{eq:time_transformed_process_inverse}for the quantity $r^{-1}(T)$. 

The estimator \eqref{eq:erg_avg_ctstime_Y_compbudg} essentially coincides with doing self-normalised importance sampling for the base process $Y$, where the unnormalised weights are $(s(Y_t))^{-1}$. The estimator \eqref{eq:erg_avg_ctstime_Y_compbudg} is considered in \citet{chak2023optimal} for a particular choice of the base process, $Y$.
Similarly to importance sampling, this approach allows the estimation of the normalising constant of $\mu$. Indeed, let $\mu(x) = Z^{-1} \exp(-\pot(x))$, where $Z$ is unknown, but where we have access to the function $\pot$. As will be seen in the proof of \Cref{thm:lln_timechange} (see \Cref{sec:proof_LLN}), we can estimate $Z$ using  
\begin{equation}
	\frac{r^{-1}(T)}{T\int_{\rmE} s(x) \exp(-\pot(x))\dd x} \,\, \xrightarrow[a.s.]{T \rightarrow \infty} \,\,  Z.
\end{equation}
This can be used when one has access to the quantity $\int_{\rmE} s(x) \exp(-\pot(x))\dd x$.

\item We now provide a third option in the case where directly simulating $X$ or the random time-change $r^{-1}(T)$ is not feasible. An alternative strategy is to use as base process $Y$ a jump process with jump rate $\lambdatilde(y) = 1$ for all $y\in\rmE$, and jump kernel $\Qtilde$, that has $\mutilde$ as stationary distribution. 
The kernel $\Qtilde$ can be any kernel from the literature on discrete time MCMC algorithms, or the transition kernel of a continuous time Markov process after some fixed time $\bar t$. In either case, the time changed process $X_t = Y_{r(t)}$ is a jump process with rate $\lambda(x) = s(x)$ and kernel $Q = \Qtilde$ (see \Cref{thm:timechange_pdmp} below), with stationary distribution $\mu$. 
Fixing a time horizon $T$, and a starting position $x$, we can directly simulate the time-changed process $X$ using the procedure described in Algorithm \ref{Markov.Jump.Process.algorithm:1}.

\begin{algorithm}\caption{Time-changed Markov jump process implementation}\label{Markov.Jump.Process.algorithm:1}
\KwIn{Time horizon $T$, starting position $x$, speed function $s$, target $\mu$, a Markov transition kernel $\Qtilde$ that leaves $\mutilde$ as in \cref{mu.tilde.def:1} invariant}  

\KwOut{Jump process $X$ targeting $\mu$.} 

Set $n \gets 0$, $T_0 = 0$, $X_0 = x$\;
\While{$T_n < T$}{
Simulate $\tau_{n+1} \sim \Exp(s(X_{T_n}))$. Set the next jumping time to be $T_{n+1} = T_n+\tau_{n+1}$\;
For all $t \in \left[T_n, \min\{T,T_{n+1}\}\right)$ set $X_t = X_{T_{n}}$\;
Simulate the next jump of the process $X_{T_{n+1}}\sim \Qtilde(X_{T_n},\cdot)$\;
}
\Return{$\left(X_t\right)_{t \in [0,T]}$.}
\end{algorithm}

This process is straightforward to simulate and has stationary distribution $\mu$ under the conditions of \Cref{thm:lln_timechange}. 
This procedure suggests the following estimator
\begin{equation}\label{eq:erg_avg_jumpprocess}
	\frac1T \int_0^T f(X_t) \dd t \,=\, 
	\frac1T \sum_{n=0}^{N_T} f(X_{T_n})(T_{n+1}-T_n) \,,
\end{equation}
where  $N_T$ is the time of the last jumping event before time $T$. One can ``regularise" the estimator \eqref{eq:erg_avg_jumpprocess} by instead discretising the path $(X_t)_{t\in[0,T]}$, that is using 
\begin{equation}\label{eq:erg_avg_jumpprocess_discrete}
	\frac1N \sum_{n=1}^{N} f(X_{n\delta}) 
\end{equation}
for some step size $\delta>0$ and an integer $N$, such that $N\delta = T$. This has the effect of ignoring with high probability the skeleton states $X_{T_n}$ where the process spends a small amount of time, that is when $T_{n+1}-T_n$ is very small. Intuitively, these are points where the target distribution is small if the speed function is chosen well.
Both estimators \eqref{eq:erg_avg_jumpprocess} and \eqref{eq:erg_avg_jumpprocess_discrete} require simulating a random number of jumps $N_T$. We can alternatively fix a computational budget, i.e. the number of jumps $N$ for the process $X$, and use the same ideas of \eqref{eq:erg_avg_jumpprocess} and \eqref{eq:erg_avg_jumpprocess_discrete}. 
In this case, \eqref{eq:erg_avg_jumpprocess} can be seen as a continuous-time counterpart of \cite{andral2023importance}, or as a Markovian counterpart of \cite{MALEFAKI20081210}.
Howerver, it is not hard to see that this is equivalent to doing self-normalised importance sampling with the discrete time Markov chain $(X_{T_n})_{n=0}^N$ with random weights, and hence there is then no practical advantage in doing this over its deterministic counterpart. 
\end{enumerate}

\begin{remark}
	All three strategies discussed above are asymptotically unbiased as $T\to\infty$ (or as $N\to\infty$), but the first two strategies can be modified to make them easier to implement at the cost of introducing a bias. In particular, in the first strategy we can instead discretise $X$, while in the second strategy, given a path of $Y$, we can approximate the one-dimensional integral $r^{-1}(t)=\int_0^t (s(Y_u))^{-1}\dd u$ with a numerical solver. 
\end{remark}

\begin{remark}
   Equation \eqref{eq:lln_equivalence} reveals a connection between sampling using time-changed processes and importance sampling. In importance sampling, samples from a measure $\tilde{\mu}$ are re-weighted by $\mu/\tilde{\mu}$ to approximate expectations with respect to $\mu$. Similarly, the paths of the time-changed process $X$ are obtained by re-weighting paths from the $\tilde{\mu}$-stationary base process, where $\tilde{\mu}$ is defined in \eqref{mu.tilde.def:1}. Here, the time change assigns importance weights, where large weights are assigned to states where the speed function is small, i.e. where the time-changed process moves slowly. The speed function, defined as $\tilde{\mu}/\mu$ up to a constant factor, can then be interpreted as the reciprocal of the importance weight. Intuitively, a well-chosen speed function defines a distribution $\tilde{\mu}$ that makes the task of interest easier to solve, e.g. connecting the modes of $\mu$ when it is multimodal. As in importance sampling, it is essential to choose speed functions bounded below by a strictly positive constant, ensuring that $\tilde{\mu}$ has heavier tails than $\mu$. 

\end{remark}

\section{Examples}\label{sec:examples}

\subsection{Time-changes of piecewise deterministic Markov processes (PDMPs)}\label{sec:pdmp}

In this section we show that time-changes of general PDMPs are also PDMPs and we identify their characteristics.
The PDMPs considered in the MCMC literature are typically of the form $(X,V)$, where $X$ is the position vector and $V$ can be interpreted as a velocity vector. With an abuse of notation, we denote the stationary distribution of the PDMP as $\mu(\dd x,\dd v) = \mu(\dd x) \nu(\dd v)$, where $\mu(\dd x)$ is the target distribution and $\nu$ is an auxiliary distribution. In particular, the position vector is the object of interest in the context of MCMC since it is marginally invariant with respect to the target distribution.
A PDMP is identified by the triple $(\vf,\lambda,Q)$, which governs the dynamics of the process as follows:
\begin{itemize}[leftmargin=5mm]
    \item The smooth vector field $\vf$ defines the deterministic motion of the process through the system of ODEs $\dd \varphi_t(y,w)=\vf(\varphi_t(y,w)) \dd t$
    where $\varphi_t$ denotes the integral curve that solves the ODE, and satisfies $\varphi_0(y,w)=(y,w)$
    . The process starting from $(y,w)$ follows the deterministic path $\left(\varphi_t(y,w)\right)_{t \geq 0}$.
    \item Noise comes into the system at random times $\tau$ following the probability distribution $$\mathbb{P}_{y,w}(\tau>t)=\exp\left(-\int_0^t \lambda(\varphi_u(y,w))\dd u\right),$$ where $\lambda:\rmE\to [0,\infty)$ is the event rate and $(y,w)$ is the initial condition. 
    \item At these random times the state of the PDMP is updated by drawing from the probability kernel $Q(\varphi_{\tau}(y,w),\dd (y',w'))$ 
    .
\end{itemize}

Now we connect the characteristics of the base PDMP with its time change.
\begin{theorem}[Time change of a PDMP]\label{thm:timechange_pdmp}
	Assume that \Cref{ass:s.integrability} holds. Consider a PDMP with characteristics $(\vftilde,\lambdatilde,\Qtilde)$. Then the time-changed process with speed function $s$ is the PDMP with characteristics $(\vf,\lambda,Q)=(s\vftilde,s\lambdatilde,\Qtilde)$.
\end{theorem}
\begin{proof}
    The proof is presented in \Cref{sec:proof_timechange_pdmp}.
\end{proof}
Importantly, the vector field $\vf$ and rate $\lambda$ are obtained multiplying the ones of the base PDMP with the speed function $s$, while the jump kernel $Q$ is not affected and remains equal to $\Qtilde$. 

\subsubsection{Running example: the Zig-Zag process}\label{sec:intro_zigzag}

The Zig-Zag process (ZZP) \citep{ZZ} is a PDMP $(Y,W)$ with state space $\rmE=\mathbb{R}^d\times \{-1,+1 \}^d$. 
The vector $Y\in \mathbb{R}^d$ denotes the position of the ZZP, while $W \in \{-1,+1 \}^d$ represents its velocity. 
The ZZP is designed so that it has a unique stationary distribution with density $\mutilde(y,w)  = \mutilde(y) \times  2^{-d}$ on $\rmE$, where $\mutilde(y) \propto \exp(-\pottilde(y))$ is the density of a given probability measure $\mutilde(\dd y) = \mutilde(y) \dd y$ defined on $\mathbb{R}^d$.
This process evolves according to simple deterministic dynamics between random event times, at which the process changes direction. The deterministic dynamics follow the system of ODEs, $\dd Y_t = W_t \,\dd t,   \dd W_t = 0,$ i.e. the process evolves with constant velocity.  
The ZZP has $d$ types of random events, each consisting in a change of sign of the corresponding coordinate of the velocity vector. The jump mechanism for the coordinate $i \in \{ 1,\dots , d \}$ is described by the following pair of jump rate and kernel
\begin{equation}\label{eq:switching_rates_zzp}
\lambdatilde_i(y,w)=(w_i\partial_i \pottilde(y))_+ +\gammatilde_i(y), \quad \Qtilde_i((y,w),(\dd \tilde{y},\tilde{w})) = \delta_{(y,R_iw)}(\dd \tilde{y},\tilde{ w})
\end{equation}
\sloppy for all $(y,w),(\tilde y, \tilde w)\in\rmE$, where $\gammatilde_i(y)\geq 0$ are refreshment rates that can be set to zero and $R_i$ is the map $w \mapsto (w_1,\dots, w_{i-1},-w_i,w_{i+1},\dots, w_d)$. Here we used the notation $a_+ = \max\{0,a\}$, and $\partial_i$ denotes the partial derivative with respect to $y_i$.

\begin{remark}
    Notice that PDMPs with more than one type of jump, each with their own rate and kernel $(\tilde\lambda_i,\Qtilde_i)_{i=1,\dots,n}$, can be described in the formulation of \Cref{sec:pdmp} by setting $\lambdatilde = \sum_{i=1}^n \lambdatilde_i$ and $\Qtilde = \sum_{i=1}^n (\nicefrac{\lambdatilde_i}{\lambdatilde})\, \Qtilde_i$. This has the interpretation that event times are governed by the total rate $\lambdatilde$, and at each event time there is probability $\nicefrac{\lambdatilde_i}{\lambdatilde}$ of obtaining the new state applying $\Qtilde_i$.
\end{remark}

Now we characterise the time change of the ZZP, enforcing that the resulting process has the target distribution $\mu\propto\exp(-\pot)$ as marginal stationary measure for the position vector. As given by \Cref{thm:lln_timechange}, this can be achieved considering as base process a standard ZZP with marginal invariant distribution $\mutilde \propto s \mu$, that is the distribution with density $\mutilde(y)\propto \exp(-\pottilde(y))$ for $\pottilde (y) = \pot(y) - \ln s(y).$
We shall assume the following condition on $s$ and $\pot$ to ensure that the base ZZP satisfies a LLN \citep{Bierkensergodicity}.
\begin{assumption}\label{ass:growth_ergo_ZZ}
    The function $\pottilde (y) = \pot(y) - \ln s(y)$ is $\mathcal{C}^3(\R^d)$, has a non-degenerate local minimum, and there exist constants $c>d$, $c'$ such that $\pottilde(y) >c \ln(\lvert y\rvert)-c'$ for all $y\in\R^d.$
\end{assumption}
In the following result, we restrict our attention to speed functions that only depend on the position vector.
\begin{proposition}[Time-changed ZZP]\label{prop:suzz}
	Let $s:(x,v)\mapsto s(x)$ be a speed function such that $s\in\mathcal{C}^1(\R^d)$.
    Suppose that \Cref{ass:s.integrability} and \Cref{ass:growth_ergo_ZZ} are satisfied.
    Let $(Y,W)$ be a ZZP with invariant density $\mutilde(y,w)\propto\exp(-\pottilde(y))$.  
	Then, the process $(X_t,V_t) = (Y_{r(t)},W_{r(t)}) \text{ for } r(t)= \int_0^t s(X_u) \dd u$
	is a PDMP with deterministic dynamics 
	 \begin{equation}\label{eq:ODE_suzz}
		\begin{aligned}
			& \dd X_t = s(X_t)V_t\, \dd t,\quad \dd V_t = 0,
		\end{aligned}
	\end{equation}
	and with jump mechanism described by the rates and kernels for $ i=1,\dots,d,$
	\begin{equation}\label{eq:switching_rates_suzz}
		\lambda_i(x,v)=\left(v_i\left(s(x) \partial_i \pot(x) - \partial_i s(x) \right)\right)_+ \,+s(x)\gammatilde_i(x), \quad Q_i ((x,w),(\dd \tilde{x},\tilde{v})) = \delta_{(x,R_iv)}(\dd \tilde{x},\tilde{v}).
	\end{equation}
	Moreover, $(X_t,V_t)$ has unique stationary density $\mu(x,v)= \mu(x)\times  2^{-d}$ for $\mu\propto\exp(-\pot)$, and it satisfies a LLN.
\end{proposition}
\begin{proof}
    The statement on the dynamics, jump rates and kernels of the process is a corollary of \Cref{thm:timechange_pdmp}.
	The statement on the stationary distribution and LLN follows from \Cref{thm:lln_timechange} along with the fact that the ZZP satisfies a LLN itself, as proven in \citet{Bierkensergodicity}.
\end{proof}

It turns out the very same process described in \Cref{prop:suzz} was introduced by \cite{Vasdekis_speedup} with the goal of improving convergence for heavy tailed distributions. In \cite{Vasdekis_speedup} the process, termed \emph{Speed Up Zig-Zag}, was obtained with ad-hoc arguments and not as a time changed ZZP, a point of view that allows to take advantage of known properties of the base process, i.e. the ZZP. Moreover, seeing the $(X,V)$ as a time-change of a ZZP with invariant distribution $\mutilde$ gives a clear understanding of the role of the speed function on the dynamics of the process.

\subsubsection*{One-dimensional bimodal target}\label{sec:bimodal_example}
\begin{figure}[t]
\begin{subfigure}[t]{\textwidth}
		\centering
		\includegraphics[width=0.35\textwidth]{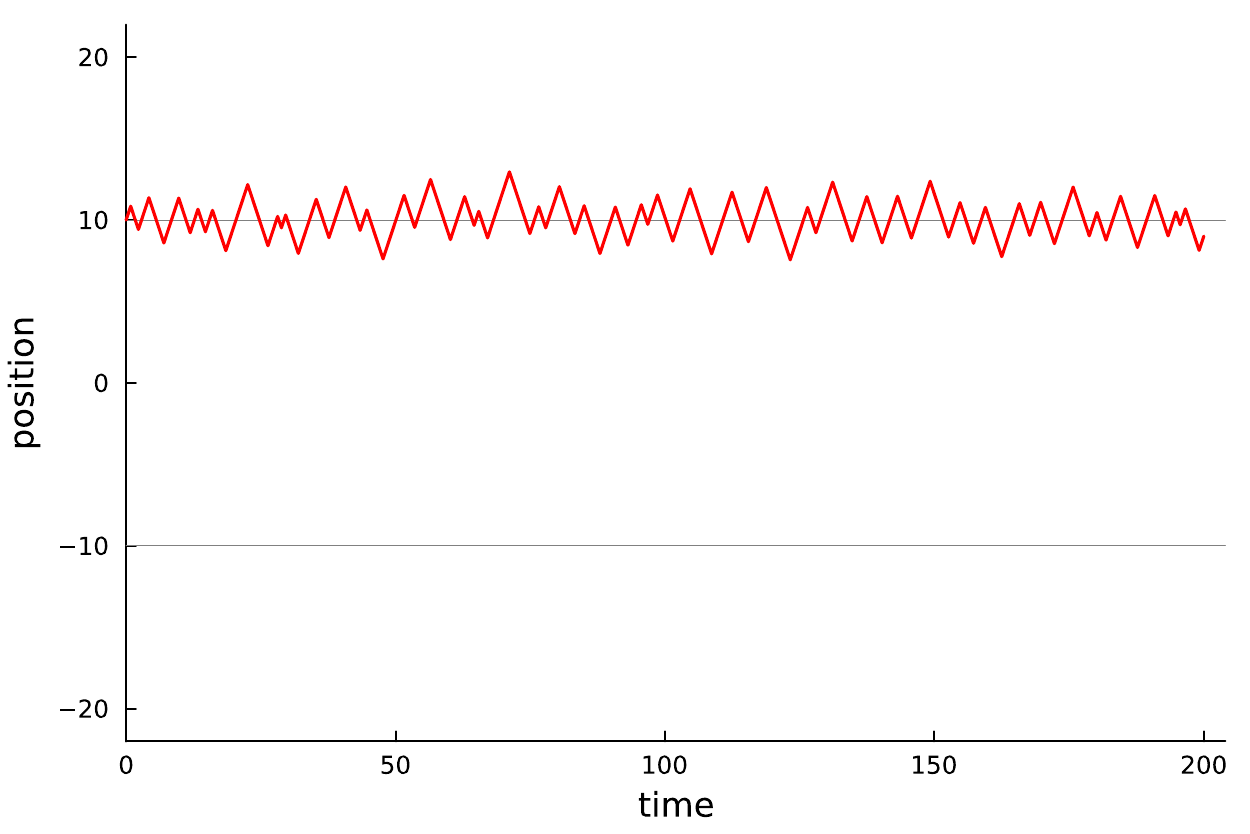}
        \hspace{20pt}
        \includegraphics[width=0.35\textwidth]{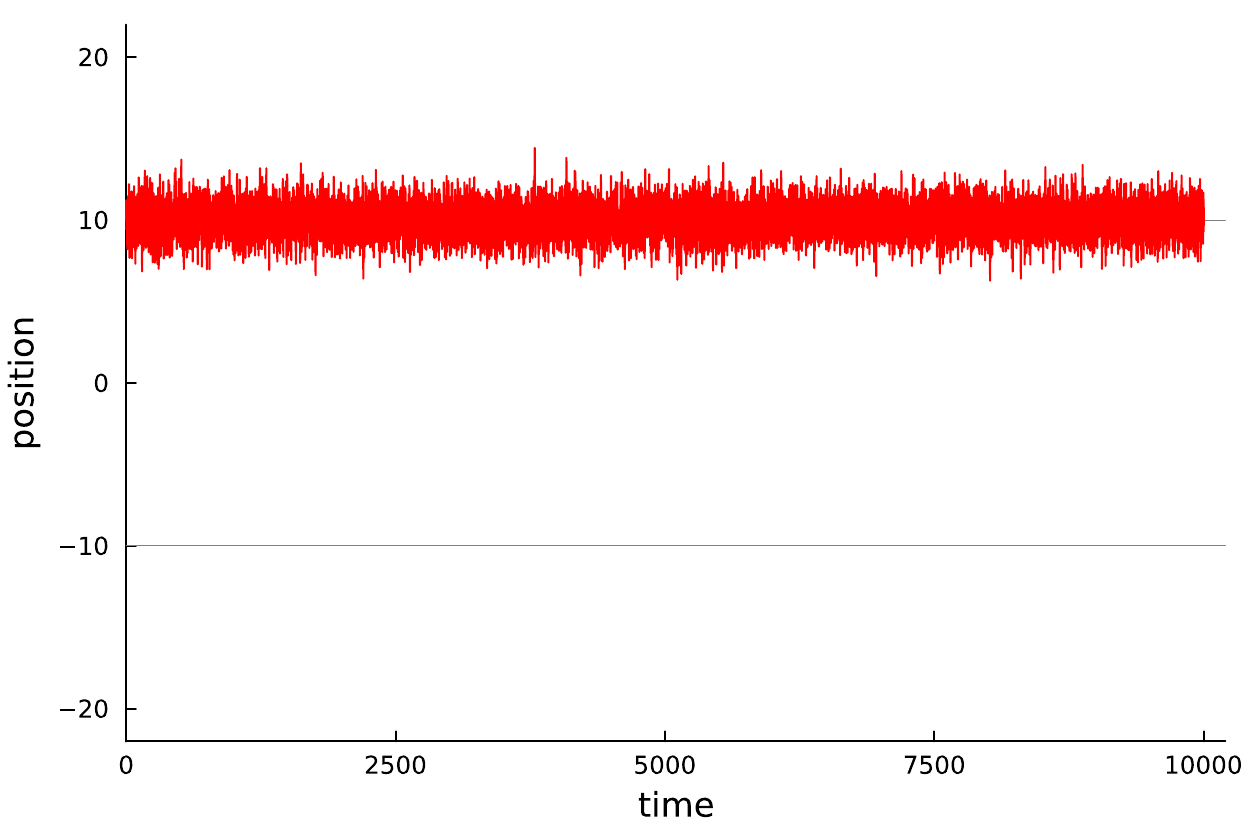}
		\caption{Results for the speed function $s(x)=1$.}
	\end{subfigure}
    \begin{subfigure}[t]{\textwidth}
		\centering
		\includegraphics[width=0.35\textwidth]{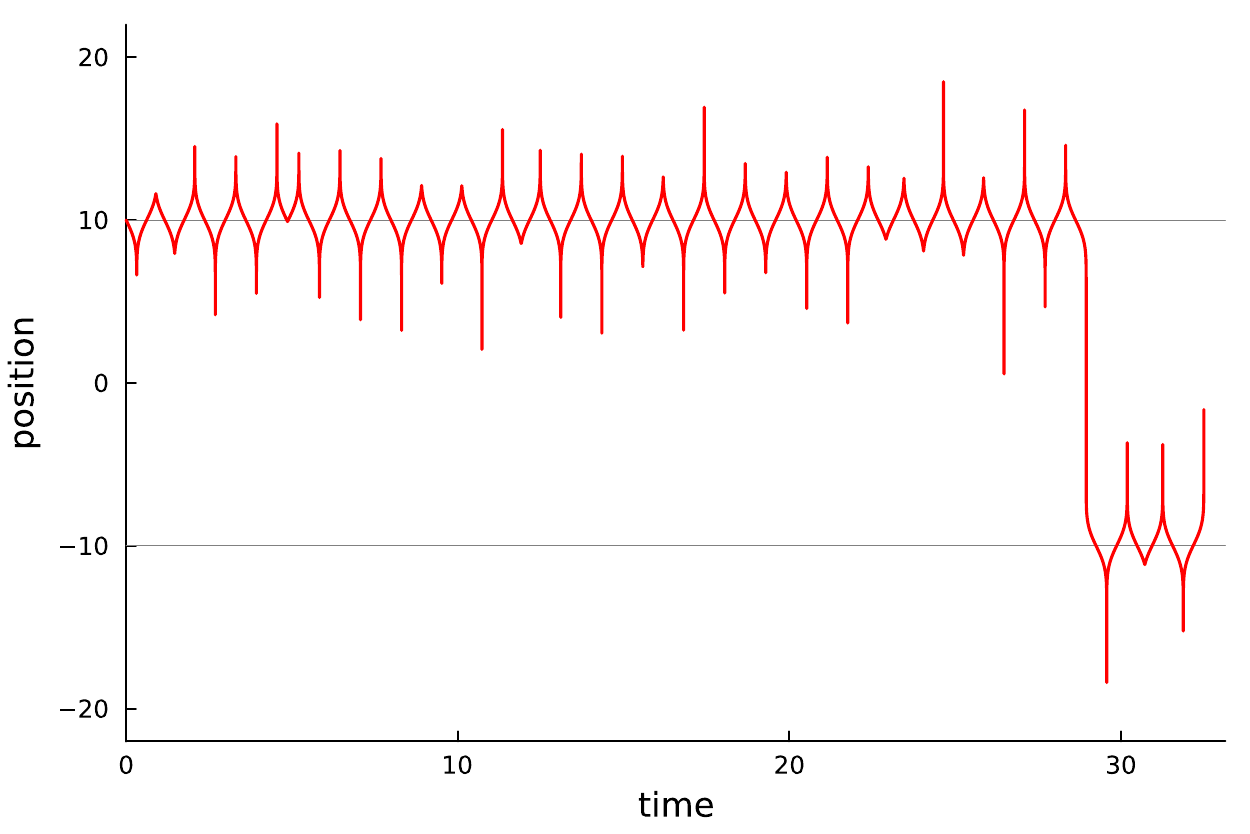}
        \hspace{20pt}
        \includegraphics[width=0.35\textwidth]{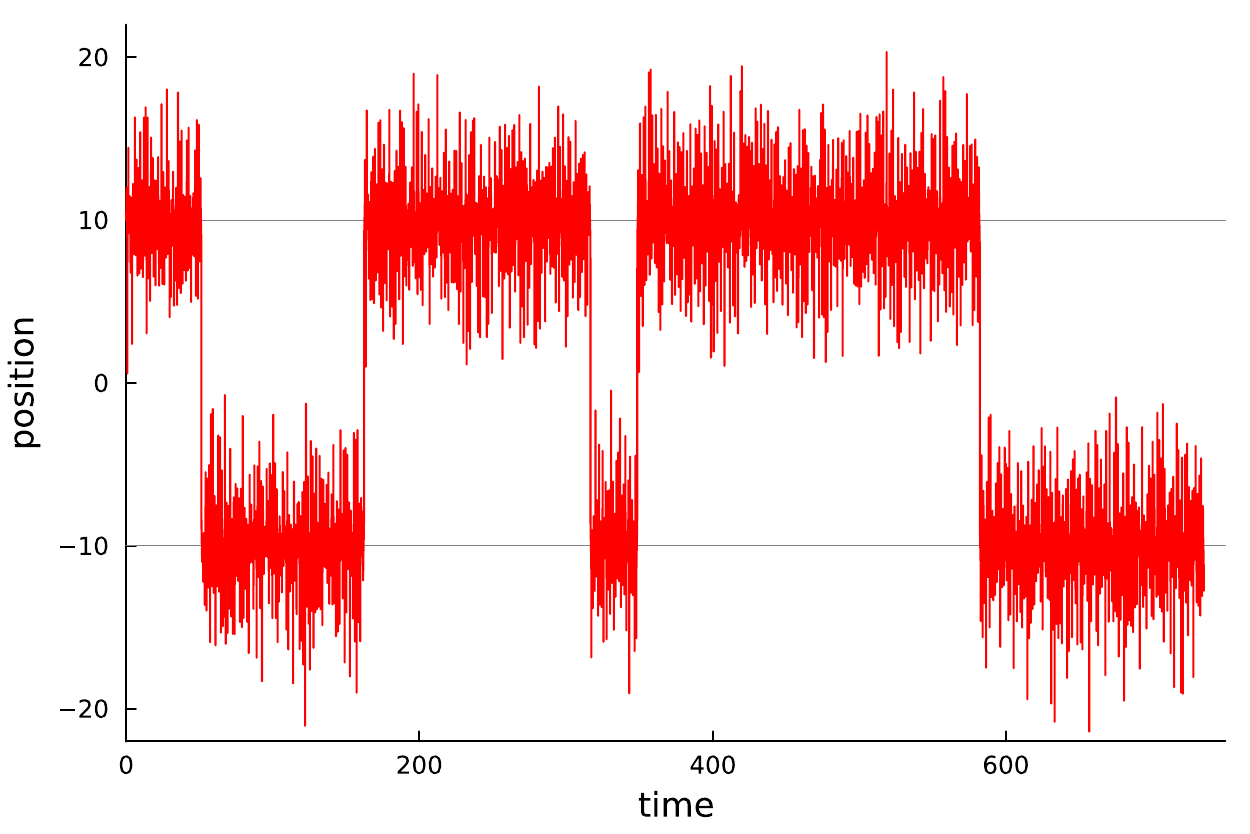}
		\caption{Results for the speed function $s(x)=\mu^{-\alpha}(x)$ with $\alpha = 0.9$.}
	\end{subfigure}
	\caption{Simulations for a mixture of two, one-dimensional standard normal random variables with means at $10$ and $-10$, and equal weight. The base process is the standard ZZP. The plots on the left show short simulations, highlighting the different dynamics, while the plots on the right show longer runs of the processes.}
	\label{fig:gaussianmixture_1d}
\end{figure}
Consider the case where $\mu$ is a one-dimensional bimodal density $\mu(x) \propto \exp(  -\pot(x) )$, corresponding to a smooth double well potential $\pot$ with two minima at $x_0, x_2 \in \mathbb{R}$ and a local maximum at $x_1 \in \mathbb{R}$, where $x_0 < x_1 < x_2$.  Here we compare the probabilities that, starting at $x_0$ with velocity $+1$, the ZZP and its time changed counterpart successfully reach the mode $x_2$ in one go, that is without intermediate velocity flips. This setting was studied for the ZZP by \cite{Monmarche2016}. Let us constraint on speed functions such that $\sign(\pottilde'(x)) = \sign(\pot'(x))$ for $\pottilde(x) = \pot(x) -\ln s(x)$ and assume the refreshment rates are zero, i.e. $\gamma \equiv 0$. The probability that the time changed ZZP, $(X,V)$, reaches the mode around $x_2$ is equal to the probability that the base process, $(Y,W)$, has its first event after it has reached the point $x_1$, since the switching rate is zero when $x_1<x<x_2$ and the velocity is $V_0=+1$. Defining $\tau_1 :=\inf\{ t>0: V_t=-1\},$  with simple computations we can show that 
\begin{equation}\label{eq:prob_crossing_suzz}
	\begin{aligned}
		\mathbb{P}_{(x_0,+1)}(X_{\tau_1}\geq x_1) &=
		\frac{s(x_1)}{s(x_0)} \exp( -\left(\pot(x_1)-\pot(x_0)\right)).
	\end{aligned}
\end{equation}
Recalling that the standard ZZP corresponds to the choice $s(x)=1$, we see that the probability of crossing the low density region for the time changed ZZP is increased if $s(x_1)>s(x_0)$ and decreased if $s(x_1)<s(x_0)$. This confirms our intuition on the benefits of increasing the speed in low-density regions that constitute a bottleneck for the base process. As a result, the time-changed process with speed satisfying $s(x_1)>s(x_0)$, and similarly $s(x_1)>s(x_2)$, will hop between the two modes more often than the base process. 
A natural choice in this case is $s(x)=\exp(\alpha \pot(x))$ with $\alpha\in (0,1)$, where values $\alpha \geq 1$ are not allowed because the time-changed process becomes explosive \citep{Vasdekis_speedup}. We illustrate this choice of speed function in \Cref{fig:gaussianmixture_1d}. In this case  $\mathbb{P}_{(x_0,+1)}(X_{\tau_1}\geq x_1) = \exp( -(1-\alpha)(\pot(x_1)-\pot(x_0))),$ effectively corresponding to using $Y$ in order to target $\mu^{1-\alpha}$, a tempered version of $\mu$.
In \Cref{sec:eyring-kramers} we give an Eyring-Kramers formula for the time changed ZZP, i.e. a bound on the expected hitting time of point $x_1$ starting from $x_0$ with velocity $-1$, leveraging existing result for the standard ZZP developed in \cite{Monmarche2016}. 

\subsubsection*{Tail exploration}
\begin{figure}[t]
\begin{subfigure}[t]{0.49\textwidth}
		\centering
		\includegraphics[width=0.9\textwidth]{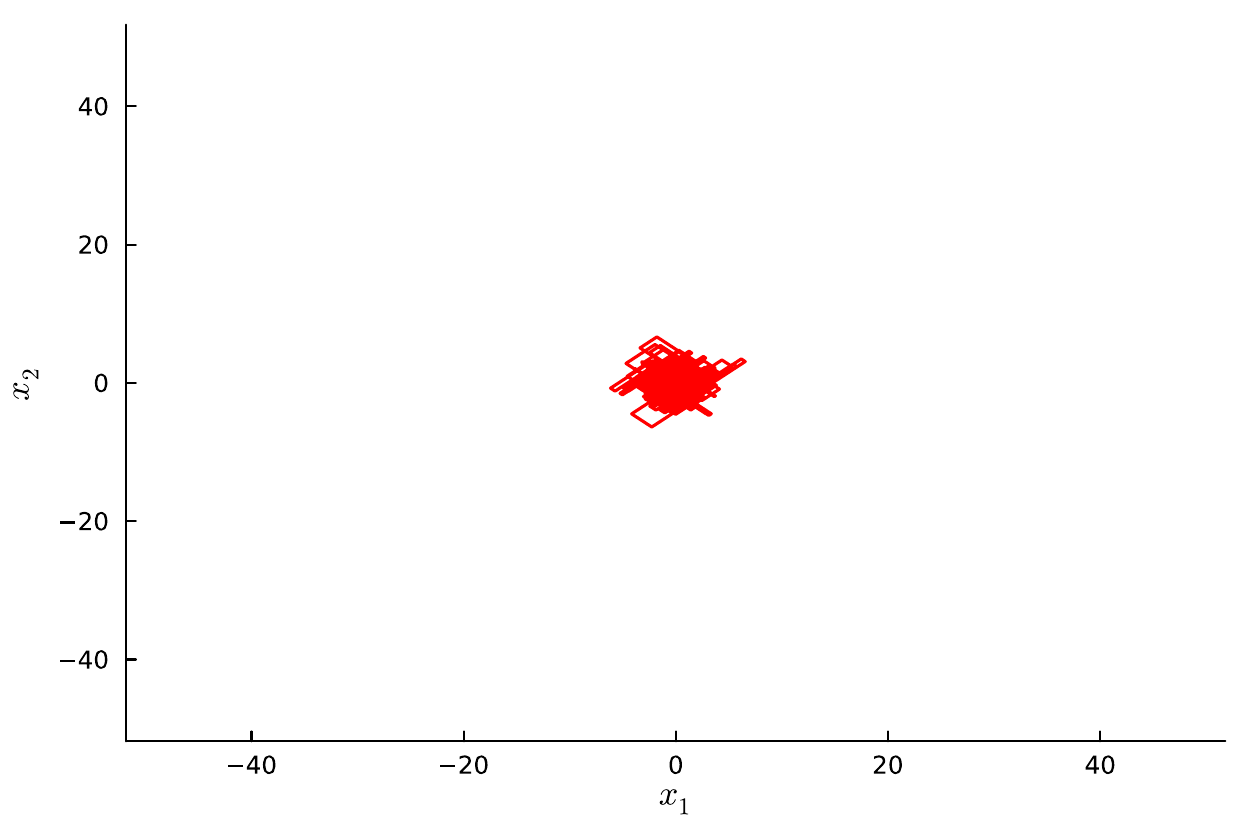}
        \caption{$s(x) = 1.$}
	\end{subfigure}
    \hfill
    \begin{subfigure}[t]{0.49\textwidth}
    \centering
        \includegraphics[width=0.9\textwidth]{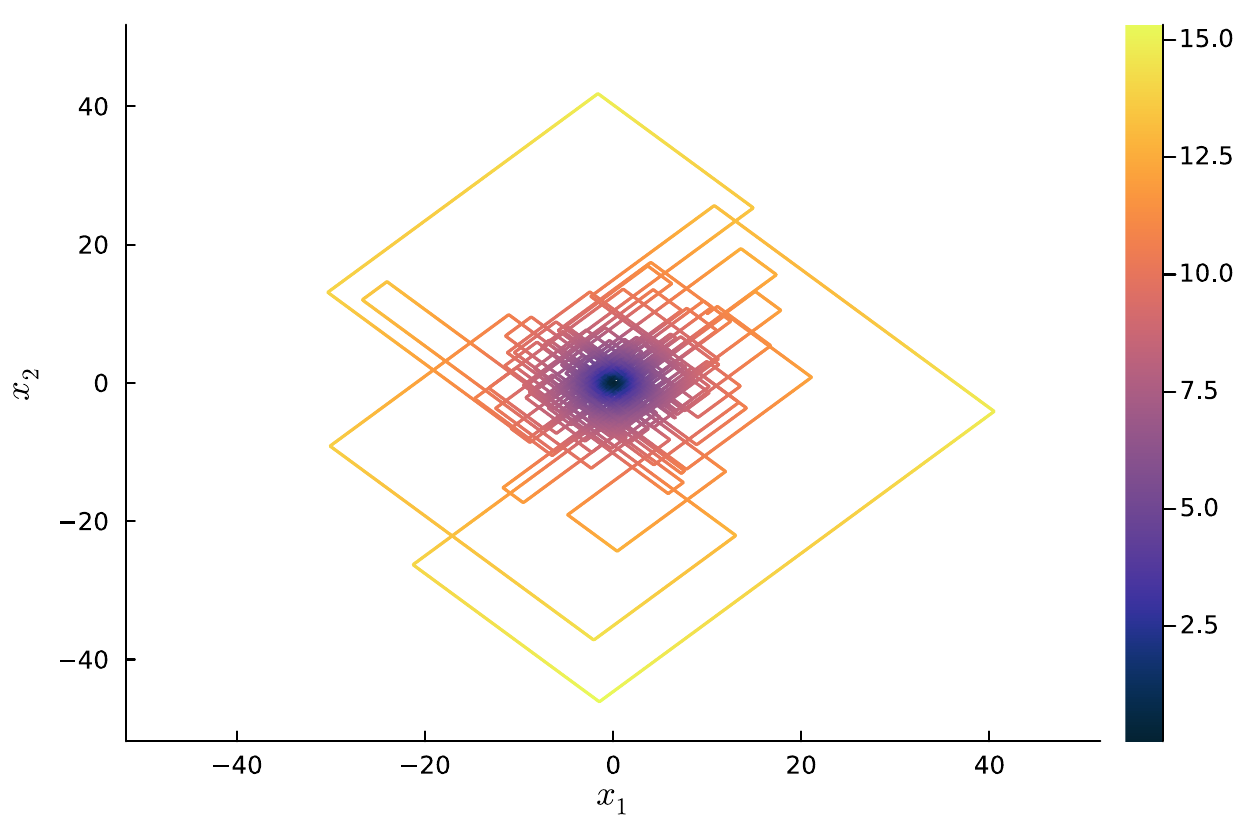}
		\caption{$s(x)=(1+\lvert x\rvert^2)^2$.}
    \end{subfigure}
	\caption{Simulations for a two-dimensional, isotropic student distribution with mean at the origin and $5$ degrees of freedom. The plot on the right shows the logarithm of the velocity in the colourmap.}
	\label{fig:heavytailed_2d}
\end{figure}
Consider a $d$-dimensional unimodal target such that $\pot$ is a differentiable function with global minimum at $x_0$. We study the probability of reaching the tails of the distribution, a task of particular interest when one is interested in estimating expectations of observables that depend on tail regions, such as the probability of rare events.
We consider the time changed ZZP with speed $s$, initial position $x_0$ and velocity $v$ and compute the probability that it reaches a certain distance from $x_0$ before the first even $\tau_1$ takes place. For any $c>0$, the process reaches a distance greater than $c$ with the same probability of a standard ZZP with target $\mutilde$, and therefore we obtain 
\begin{equation}\notag
	\begin{aligned}
		\mathbb{P}_{(x_0,v)}(\lvert X_\tau - x_0 \rvert \geq c) &= \frac{s(x_0+v  c/\sqrt{d})}{s(x_0)}\exp\left(-(\pot(x_0+v  c/\sqrt{d})-\pot(x_0))\right). 
	\end{aligned}
\end{equation}
The process travels a larger distance before the first event takes place if $s$ increases in the tails. Examples of suitable speed functions in this context are $s(x) = 1 + \lvert x-x_0\rvert^2$, with deterministic dynamics that can be implemented directly without numerical error (see \cite{Vasdekis_speedup}), or $s(x) = 1+\pot(x)$ assuming $\pot$ is positive everywhere, but also $s(x) = \exp(\alpha\pot(x))$.

\subsubsection{Other piecewise deterministic Markov processes}\label{sec:PDMP_TT}
Taking advantage of \Cref{thm:timechange_pdmp}, we now obtain several novel PDMPs with 
stationary distribution $\mu$.
Using the same ideas as in \Cref{sec:intro_zigzag}, it is possible to define time-changes of PDMPs such as those designed for subsampling (see e.g. \cite{ZZ}), or of other PDMPs such as \citet{coordinate_sampler,bierkens2020boomerang}.
\begin{example}[Bouncy Particle Sampler (BPS) \citep{BPS}]\label{ex:bps_tt}
	 The BPS is a PDMP $(Y,W)$ defined on $\rmE=\mathbb{R}^d\times \mathbb{R}^d$ , with position-velocity decomposition and with deterministic dynamics identified by the vector field $\vftilde_{\rmBP}(y,w)=(w,0)$. We now describe the standard BPS with stationary distribution with density $\mutilde(y,w)\propto \exp(-\pottilde(y))\nu(w)$, where $\nu$ is a  $d$-dimensional standard normal distribution.
	 The event rate of BPS is $\lambdatilde_{\rmBP}(y,w) = \langle w,\nabla \pottilde(y)\rangle_+ + \lambda_r$, where $\lambda_r>0$ is the refreshment rate, and $\langle \cdot , \cdot \rangle$ denotes the usual inner product in $\mathbb{R}^d$. The jump mechanism  is given by
	 \begin{equation}\notag
		 	\Qtilde_{\rmBP} ((y,w),(\dd y',\dd w')) =\frac{\langle w,\nabla \pottilde(y)\rangle_+}{\lambdatilde_{\rmBP}(y,w)} \, \delta_{(x,\Rtilde_{\rmBP}(x)w)}(\dd y',\dd w') + \frac{\lambda_r}{\lambdatilde_{\rmBP}(y,w)} \,  \delta_x(\dd y')\nu(\dd w'),
	 \end{equation}
	 where $\Rtilde_{\rmBP}$ is the operator that defines elastic reflections on the level curves of the potential $\pottilde$, that is $\Rtilde_{\rmBP}(y)w=w-2\frac{\langle w,\nabla \pottilde(y)\rangle}{\lvert\nabla \pottilde(y)\rvert^2} \nabla \pottilde(y).$
	At event times, the process is either reflected according to $\Rtilde_{\rmBP}$, or the velocity is refreshed by an independent draw from $\nu$.
	 \citet{BPS} prove that if $\lambda_r >0$ and $\pottilde \in \mathcal{C}^1$, the process has $\mutilde$ as stationary distribution and a LLN (i.e. \Cref{ass:lln_base}) is satisfied. 
    
	Applying \Cref{thm:timechange_pdmp} we find that the time changed BPS, $(X_t,V_t)=(Y_{r(t)},W_{r(t)})$, with speed $s$ satisfying \Cref{ass:s.integrability}, is the PDMP with characteristics $(s\vftilde_{\rmBP},s \lambdatilde_{\rmBP},\Qtilde_{\rmBP})$. \Cref{thm:lln_timechange} gives that, if $\lambda_r >0$ and $\pottilde , s \in \mathcal{C}^1(\R^d)$, this process has invariant distribution with density $\mu(x,v)= \mu(x)\nu(v)$ for $\mu(x) \propto\exp(-\pot(x))$ when $\pottilde = \pot - \ln s $. Simple calculations show that the time-changed BPS has event rate $\lambda(x,v) = \langle v,s(x) \nabla \pot(x)-\nabla s(x) \rangle_+ + \lambda_r s(x)$ and deterministic motion given by \eqref{eq:ODE_suzz}. We note that the jump kernel $Q$ consists of bounces on the contour lines of $\mutilde$. 
\end{example}

\begin{example}[Randomised Hamiltonian Monte Carlo (RHMC) \citep{bou2017randomized}]\label{ex:randomisedHMC}
	RHMC is a PDMP which follows Hamiltonian dynamics with momentum refreshments at random times.
	The RHMC with stationary distribution $\mutilde(y,w) \propto \exp(-\pottilde(y)-(1/2)\lvert w\rvert^2)$ has dynamics described by the system of ODEs 
	\begin{equation}
			\frac{\dd Y_t}{\dd t} = W_t, \quad \frac{\dd W_t}{\dd t} = -\nabla \pottilde(Y_t), 
	\end{equation}
	while at random times with rate $\lambda_r$ the velocity is refreshed with a draw from the standard normal distribution.
	
	The time-changed RHMC with speed $s$ has unique invariant distribution with density $\mu(x,v) \propto \exp(-\pot(x)-(1/2)\lvert v\rvert^2)$ for $\pottilde = \pot - \ln s$. This process is a PDMP with deterministic dynamics
	\begin{equation}\notag
		\begin{aligned}
			& \frac{\dd X_t}{\dd t} = s(X_t)V_t, \quad \frac{\dd V_t}{\dd t} = -s(X_t) \nabla \pot(X_t) + \nabla s(X_t). 
		\end{aligned}
	\end{equation}
	Such process moves on the level curves of the Hamiltonian $H(x,v)=\pot(x)-\ln s(x)+(1/2)\lvert v\rvert^2$ with speed given by $s(x)$. The jump mechanism consists again of velocity refreshments from the standard normal distribution with rate $s(x)\lambda_r$. With analogous ideas it is possible to define a time-changed version of RHMC samplers with different choices kinetic energy, considered e.g. \cite{Hamiltonian_ZZS}.
\end{example}

\subsection{Diffusion processes}\label{sec:TT_diffusions}
In this section we apply the framework of time-changes to Langevin diffusions. We start with a statement characterising the time change of a diffusion process with characteristics $(\btilde,\sigmatilde)$, that is the process satisfying the SDE $\dd Y_t = \btilde (Y_t)\dd t + \sigmatilde(Y_t)\dd B_t$. where $\btilde:\R^d\to\R^d$ is the drift coefficient, $\sigmatilde:\R^{d}\to \R^{d\times d}$ is the diffusion coefficient, and $B_t$ is a Brownian motion.

\begin{proposition}[Time change of a diffusion process]\label{prop:timechange_diffusion}
	Consider a speed function $s$ and a diffusion with characteristics $(\btilde,\sigmatilde)$. The time changed process defined through \eqref{eq:time_change} with speed $s$ is the diffusion with characteristics $(b,\sigma)=(s\btilde,\sqrt{s}\sigmatilde)$.
\end{proposition}
\begin{proof}[Proof of \Cref{prop:timechange_diffusion}]
The proof is presented in Appendix \ref{proof.diffusion.time.change}
\end{proof}

\begin{example}[Overdamped Langevin diffusion]\label{ex:overdamped}
	The overdamped Langevin diffusion with stationary distribution $\mutilde\propto \exp(-\pottilde)$ for $\pottilde = \pot - \ln s$ is the solution of the following stochastic differential equation (SDE):
	\begin{equation}\notag
		\dd Y_t = \left(-\nabla \pot(Y_t) +\frac{\nabla s(Y_t)}{s(Y_t)} \right)\dd t + \sqrt{2}\, \dd B_t,
	\end{equation}
	This SDE was considered in \cite{chak2023optimal} as a basis for the importance sampling procedure described in \eqref{eq:erg_avg_ctstime_Y_compbudg}. The search for the optimal proposal distribution studied in \cite{chak2023optimal} is then equivalent to finding the optimal time-change of the overdamped Langevin diffusion.
	
	Applying \Cref{prop:timechange_diffusion} we find that the process $X_{t} = Y_{r(t)}$ is the solution of the SDE
	\begin{equation}\label{eq:timechange_overdamped}
		\dd X_t = \left(-s(X_t)\nabla \pot(X_t) +\nabla s(X_t) \right)\dd t + \sqrt{2 s(X_t)}\, \dd B_t,
	\end{equation}
	and has invariant measure with density $\mu(x) \propto \exp(-\pot(x))$.
	For $s(x)=\exp(\alpha\pot(x))$ with $\alpha \in(0,1)$, the process $X$ coincides with the one considered in \cite{RobertsStramer}. This choice of speed function was shown to be optimal in some settings \citep{lelièvre2024optimizingdiffusioncoefficientoverdamped}. 
    We remark also that \citet{Roberts_Ros_Peskunordering} obtained a Peskun ordering for SDEs of the form \eqref{eq:timechange_overdamped}, showing that the asymptotic variance of the diffusion \eqref{eq:timechange_overdamped} with speed $s_1$ is smaller than with speed $s_2$ when $s_1(x)>s_2(x)$ for all $x\in\rmE.$
\end{example}

\begin{example}[Underdamped Langevin diffusion]\label{ex:underdamped}
	Consider the process $(Y,W)$ solving the SDE
	\begin{align}
		&\dd Y_t = W_t \dd t, \quad \dd W_t = - \nabla \pottilde(Y_t) \dd t - W_t \dd t  + \sqrt{2} \,\dd B_t,
	\end{align}
    where $\pottilde=\pot - \ln s$. 
	This process has $\mutilde(y,w)\propto \exp(-\pottilde(y) -(1/2) \lvert w\rvert^2 )$ as stationary distribution under mild assumptions. 
	
	Following \Cref{prop:timechange_diffusion} we find that the time-change $(X_t,V_t)=(Y_{r(t)},W_{r(t)})$ solves the SDE
	\begin{equation}\notag
		\begin{aligned}
			&\dd X_t = s(X_t) V_t \dd t\, \quad  \dd V_t = - (s(X_t) \nabla \pot(X_t)-\nabla s(X_t)) \dd t  -s(X_t)  V_t \dd t + \sqrt{2 s(X_t)}\, \dd B_t.
		\end{aligned}
	\end{equation}
	Assuming this SDE is well posed, we find that the solution has the measure with density $\mu(x,v) \propto \exp(-\pot(x)-(1/2)\lvert v \rvert^2)$ as stationary distribution. \citet{Leimkuhler2024} consider the well-posedness of this time-changed SDE (as well as that of \Cref{ex:overdamped}) and propose discretisation schemes based on splitting schemes.
\end{example}

\subsection{Jump processes on discrete state spaces}\label{sec:jumpproc_discrete}
Now we consider the setting of a target distribution $\mu$ with support on a discrete state space $\mathrm{E}$. The process we consider in the following example is a pure jump Markov process, i.e. it is also a PDMP with vector field $\vf=0$. Nonetheless, we find it relevant to apply our framework to a discrete process for the sake of illustration.

\begin{example}
	\cite{power2019acceleratedsamplingdiscretespaces} introduced a pure jump process that jumps at random times from a current state $x\in \mathrm{E}$ to a neighbouring state $z\in\mathcal{N}(x)$, where $\mathcal{N}(x)\subset \rmE$ denotes the set of neighbours of $x$. 
	Each transition $x \to z$ is associated with the event rate $\lambda_z(x) = g\big( \frac{\mutilde(z)}{\mutilde(x)}\big)$, for some density $\mutilde$. This process is $\mutilde$-invariant when the function $g$ satisfies the condition $g(t)=tg(1/t)$, which holds e.g. when $g(t)= 1\wedge t$ or $g(t) = \frac{t}{1+t}$ \citep{power2019acceleratedsamplingdiscretespaces}. 
	
	We now apply a time-change with speed $s$ to this process, by choosing $\mutilde\propto s \mu$ in order to obtain a $\mu$-invariant jump process.
	Applying \Cref{thm:timechange_pdmp} we find that the time-changed process is a pure jump process that jumps from $x\in\mathrm{E}$ to a neighbouring state  $z\in\mathcal{N}(x)$ with rate $s(x)g\big( \frac{\mutilde(z)}{\mutilde(x)}\big)$.
	Interestingly, the time changed process can be simulated exactly, just like the original process of \cite{power2019acceleratedsamplingdiscretespaces}. 
	The case of variations of this process on a continuous state space \citep{Livingstone2022} can be handled in the same fashion.
\end{example}

\section{Convergence properties of time changed Markov processes}\label{sec:theory}

In this section we focus on convergence properties of a time changed process $X$ under suitable assumptions on the base process $Y$.

\subsection{Ergodicity of time changed processes}\label{sec:convergence_MP}

In this section we are interested in the qualitative convergence of the law of the time-changed process, $X$, to its stationary distribution. We shall obtain conditions under which $X$ satisfies, for any $x\in\rmE$,
\begin{equation}\label{eq:geom_ergodicity}
	\lVert \mathbb{P}_x(X_t \in \cdot)-\mutilde(\cdot)\rVert_{\text{TV}} \leq \rmM(x)\,  \rho^t
\end{equation}
where $\lVert\cdot\rVert_{\text{TV}}$ is the total variation (TV) distance, $\rho\in(0,1)$ is a constant, and $\rmM:\rmE \to \R_+$. When \eqref{eq:geom_ergodicity} holds, $X$ is \emph{geometrically ergodic}, while we say it is \emph{uniformly ergodic} when $\rmM$ is a constant function. 
This type of convergence can be proved under \emph{drift conditions} for the generator of the process \citep{DownMeynTweedie}.

Next we define some necessary notions of stability (see e.g. \citet{MeynTweedie_3}). 
The process $Y$ is said \emph{$\psi$-irreducible} if there exists a non-trivial measure $\psi$ such that for any measurable set $\rmB$ with $\psi(\rmB)>0$ it holds $\mathbb{E}_y[\int_0^\infty \1_{Y_t\in \rmB}\, \dd t]>0$. A set $ C $ is \emph{petite} if there exists a probability measure $\alpha$ on $(0,\infty)$ such that for all $y\in  C $
$\int_0^\infty \mathbb{P}_y(Y_t\in \cdot\,) \,\alpha(\dd t)  \geq \eta \nu(\cdot).$  The process is called \emph{aperiodic} if there exists a petite set $C$ and $T_0 \geq 0$, such that $\mathbb{P}_x(Y_t\in C )>0$ for all $t\geq T_0$ and all $x\in C .$

We will also use the notion of the weak generator of the process $Y$, defined as follows.
Let $\cD(\cLtilde)$ be a set of measurable functions $f:\rmE \rightarrow \mathbb{R}$ and $\cLtilde$ be an operator defined on $\cD(\cLtilde)$. 
The pair $( \cLtilde , \cD(\cLtilde)  )$ is the \emph{weak generator} of the process $Y$ if for any $f \in \cD(\cLtilde)$, and any starting point $x \in \mathbb{R}^d$, the process
\begin{equation}
    M_t=f(Y_t)-\int_0^t\cLtilde f(Y_s) \dd s
\end{equation}
is a martingale with respect to the natural filtration of the process. We will write $(\cL , \cD(\cL))$ to denote the weak generator of the time-changed process $X$, which is defined in the same way.

In the following assumption we state conditions on $Y$ which will be used in \Cref{thm:ergodicity_markovprocess} to obtain geometric and uniform ergodicity of $X$. 

\begin{assumption}\label{ass:generic_drift_condition_tt}
	$Y$ is a $\psi$-irreducible Markov process with generator $(\cLtilde,\cD(\cLtilde))$ and stationary distribution $\mutilde$. Moreover, $Y$ satisfies the following conditions.
    \begin{enumerate}
	   \item There exist bounded sets $ C  \subset   D $, constants $t_0>0$, $\eta>0$,  $\varepsilon>0$, and a non-trivial measure $\nu$ on $ C $ such that for all $y\in  C $ and for any measurable set $ A $ it holds that
	   \begin{equation}\notag
		  \int_{t_0}^{t_0+\varepsilon} \mathbb{P}_y ( Y_{t} \in  A , Y_u\in   D  \textnormal{ for all } u\in [0,t]  ) \dd t\geq \eta \nu( A ).
	   \end{equation}
       \item There exist a function $\rmV:\rmE \to [1,\infty]$, with $\rmV \in \cD(\cLtilde)$, a function $\rmW:\rmE\to\R_+$, and a constant $\gamma \geq 0$ such that for all $y \in \rmE$
	   \begin{equation}\label{eq:generic_drift}
		  \cLtilde \rmV(y) \leq - \rmW(y) \rmV(y) + \gamma \mathbbm{1}_{ C }(y),
	   \end{equation}
       where $ C $ is the same set of point $(1)$.

    \end{enumerate}

\end{assumption}

\Cref{ass:generic_drift_condition_tt}(1) is enough to ensure that the bounded set $ C $ is petite for the time-changed process (see \Cref{lem:small_sets} for the proof). For this, we require that $Y$ satisfies a petite set condition while staying inside a bounded set during a finite time horizon. We note that for most local algorithms, such as our running example, and other constant velocity PDMPs, \Cref{ass:generic_drift_condition_tt}(1) is met. The function $\rmW$ appearing in  \Cref{ass:generic_drift_condition_tt}(2),
is indicative of the convergence of the process $Y$. In particular, $Y$ is geometrically ergodic when $\rmW(y)\geq \eta$ for some $\eta>0$, and uniformly ergodic when in addition $\rmV(y)\leq \beta$ for  $\beta<\infty$ \citep{DownMeynTweedie}. Informally, $Y$ has slower convergence when these conditions are not met. 

In order to present our results, we need the following technical definition, required to translate \eqref{eq:generic_drift} to a drift condition for the time-changed process.
\begin{definition}\label{defn:nice.set}
    We define  $\mathcal{A} \subset \cD(\cL) \cap \cD(\tilde\cL)$ to be the set of measurable functions $f:\rmE \rightarrow \mathbb{R}$ such that for all $x \in E$,
    \begin{equation}\label{defn:nice.set:1}
        \cL f(x) = s(x) \cLtilde f(x).
    \end{equation}
\end{definition}
As we discuss in \Cref{appendix:generator}, the set $\mathcal{A}$ typically contains a large class of functions, for which the generator of the time-changed process is indeed of the form $s\cLtilde$. Nonetheless, this is in general not enough to ensure that the function $\rmV$ appearing in \Cref{ass:generic_drift_condition_tt}(2) is in $\mathcal{A}$, as required to prove \Cref{thm:ergodicity_markovprocess}. Throughout this section, we shall then assume that $\rmV\in\mathcal{A}$.

Finally, we will also be making the following assumption on the time-changed process.

\begin{assumption}\label{ass:X_is_aperiodic}
    $X$ is aperiodic.
\end{assumption}

Aperiodicity of $X$ is crucial in order to establish its convergence to the target. We conjecture that \Cref{ass:X_is_aperiodic} holds under aperiodicity of $Y$ and some regularity conditions. Indeed, the speed function is bounded away from zero, hence the time-change cannot force the process to get stuck anywhere in the space. 
We note that the time-changed ZZP is proven to be aperiodic under mild regularity conditions on the target and the speed function \citep{Vasdekis_speedup}. We believe that similar results can be established for the other time-changed processes discussed in \Cref{sec:examples}.

We now state the main convergence result of the section.

\begin{theorem}\label{thm:ergodicity_markovprocess}
	Suppose \Cref{ass:s.integrability} holds. 
	Consider a process $Y$ that satisfies \Cref{ass:lln_base} and  \Cref{ass:generic_drift_condition_tt} for some functions $\rmV,\,\rmW$ and set $ C $. Assume further that $\rmV \in \mathcal{A}$ as in \Cref{defn:nice.set} and that $X$ satisfies \Cref{ass:X_is_aperiodic}.
	Then, $X$ is \emph{geometrically ergodic} with respect to $\mu$ if there exists $\zeta >0$ such that $\speed(x) \geq \frac{\zeta}{\rmW(x)}$ for all $x\notin  C $. In addition, $X$ is \emph{uniformly ergodic} if $\rmV$ is bounded. 
\end{theorem}
\begin{proof}
	The proof is presented in \Cref{sec:proof_ergod_MP}.
\end{proof}

\Cref{thm:ergodicity_markovprocess} gives general conditions on the speed function to obtain a time-changed process with the wanted type of convergence. However, uniform ergodicity requires a bounded Lyapunov function $\rmV$, which is typically unavailable. The next theorem shows that for PDMPs we can obtain uniform ergodicity under an inequality of the type \eqref{eq:generic_drift} where $\rmV$ is not necessarily bounded. 
\begin{theorem}[Uniform ergodicity of a PDMP]\label{thm:uniform_ergodicity_pdmps}
	Suppose \Cref{ass:s.integrability} holds. 
	Consider a PDMP with characteristics $(\vftilde,\lambdatilde,\Qtilde)$ that satisfies Assumptions \ref{ass:lln_base}, \ref{ass:generic_drift_condition_tt} for some functions $\rmV \in \mathcal{A},\,\rmW$ and set $ C $.
	Assume that the function $\bar{\rmV}(z):= \rmV(z)\left( 1+\rmV(z) \right)^{-1} \in \mathcal{A}$. Finally, assume that \Cref{ass:X_is_aperiodic} holds. If there exists $\beta >0$ such that $s(x) \geq   \nicefrac{\beta\rmV(x)}{\rmW(x)}$ for all $x\notin  C $, then the time-changed process with speed function $s$ is \emph{uniformly ergodic}.
\end{theorem}
\begin{proof}
    The proof is postponed to \Cref{sec:proof_unif_erg_pdmp}.
\end{proof}

The proof of \Cref{thm:uniform_ergodicity_pdmps} relies on the fact that the bounded function $\bar{\rmV} (z)=\rmV(z)/(1+\rmV(z))$ satisfies a drift condition of the form \eqref{eq:generic_drift} for the generator of the process $Y$.
This theorem can be applied to obtain uniform ergodicity of the time-changed PDMPs discussed in \Cref{sec:PDMP_TT}, taking advantage of the theory developed by \cite{Bierkensergodicity,BPSexp_erg,BPS_Durmus}.

\subsection{Functional central limit theorem}\label{sec:FCLT}

In this section, we obtain a functional central limit theorem for the time changed process relying on \cite[Theorem 4.3]{glynn1996liapounov}. For this we assume the base process satisfies \Cref{ass:generic_drift_condition_tt} and $s$ is such that $s(y)\rmW(y)\rmV(y)\geq 1$ for all $y$. In the theorem below we show that
\begin{equation}\label{eq:Z_fctl}
	Z_n(t) = n^{-1/2} \int_0^{nt} (g(X_u)-\mu(g))\dd u 
\end{equation}
converges weakly to $\gamma_g B$ in the Skorokhod space, where $B$ is Brownian motion. An interesting result shown in \Cref{thm:FCLT} is that the asymptotic variance of the time changed process coincides up to a factor $\mu(s) $ with that of the base process corresponding to the function $\frac{g-\mu(g)}{s}$.
\begin{theorem}\label{thm:FCLT}
	Suppose \Cref{ass:s.integrability} holds. 
	Suppose the base process $Y$  satisfies  \Cref{ass:lln_base}. Suppose also that  \Cref{ass:generic_drift_condition_tt} holds for some $\rmV \in \mathcal{A}, \rmW$ and a set $ C $. Assume that the speed function $\speed$ is such that $\speed(x)\rmW(x)\rmV(x)\geq 1$ for all $x\in\rmE$. Suppose that \Cref{ass:X_is_aperiodic} holds.
	Finally, suppose $\mu(\rmV^2)<\infty$.
	Then for any $g:\rmE\to\R$ such that $\lvert g(z)\rvert \leq \speed(z)\rmW(z)\rmV(z)$, for $Z_n(t)$ as in \eqref{eq:Z_fctl} it holds that 
	\begin{equation}\label{eq:FCLT_convergence}
		Z_n \Rightarrow \gamma_g B \qquad \text{weakly as }n\to\infty
	\end{equation}
	in the Skorokhod topology $D[0,1]$.
	In particular, $\gamma_g^2 = 2 \int_\rmE \hat g(z) \overline g(z)  \mu(\dd z) <\infty$, where $\overline g(z) =g(z)-\mu(g)$ and $\hat{g}$ is the solution to the Poisson equation $\overline g = -\cL \,\hat{g}.$

	Assume further that $\rmW(z)\rmV(z) \geq 1$ for all $z\in\rmE$ and that $\hat g \in \mathcal{A}$. Denote as $\widetilde{\gamma}^2_{f}$ the asymptotic variance of the base process for a function $f$.
	Then it holds that
	\begin{equation}\label{eq:asymptotic_variance}
		\gamma^2_g = \gammatilde^2_{\tilde g}, 
	\end{equation}
	where $\tilde g(z) =\sqrt{\mu(s)} \,\, \frac{g(z)-\mu(g)}{s(z)}$.
\end{theorem}
\begin{proof}
	The proof is presented in \Cref{sec:proof_FCLT}.
\end{proof}

\begin{remark}
	The relation \eqref{eq:asymptotic_variance} is quite similar to the expression of the variance of self-normalised IS. Indeed, given an observable $g$, the asymptotic variance of self-normalised IS with weights $s^{-1}$ is approximately equal to 
	\begin{equation}
		\gamma^{2}_{g,snis} = (\mu(s))^2\,  \PE_{\mutilde}\left[ \frac{(g(Y)-\mu(g))^2}{(s(Y))^2}\right],
	\end{equation}
	(see e.g.\citep{mcbook}), which is essentially the asymptotic variance of the Monte Carlo estimator for a distribution $\mutilde$ and observable $\mu(s)  \frac{g(y)-\mu(g)}{s(y)}$.
	This is clearly similar to \eqref{eq:asymptotic_variance}, which has the same interpretation, but where the asymptotic variance is in terms of that of a base process with target distribution $\mutilde$.
	We further note that given an observable $g$, the self-normalised IS has optimal proposal distribution $\mutilde$, constructed by choosing inverse weights $s(x) = \lvert g(x)-\mu(g)\rvert$. Given the relation between time-change and IS, it is reasonable to expect that a similar speed will be optimal in the setting of time-changes. Naturally this choice of $s$ depends on the unknown, $\mu(g)$ and is of limited applicability.
\end{remark}

\subsection{Application to the running example}
Here we apply the results of the previous two sections to our running example. We first focus on the qualitative convergence properties of the time-changed ZZP.

We work under the following assumption on the target distribution $\mu$ and on the speed function $s$. 
\begin{assumption}\label{ass:ZZ_growth_proposal}
	Let $\mu(x,v) = Z^{-1} \exp(-\pot(x)) 2^{-d}$ for $(x,v)\in\R^d\times \{\pm 1\}^d$ and let $s\in\mathcal{C}^2(\R^d)$ be a speed function satisfying \Cref{ass:s.integrability}. Assume that the function $\rmV$, defined in Equation \eqref{Lyapunov.for.ZZ} in \Cref{appendix:ZZproofs}, is an element of $\mathcal{A}$ as in \Cref{defn:nice.set}.
	Assume further that the following conditions hold.
	\begin{enumerate}
		\item Let $\pottilde(x) = \pot(x) -\ln s(x)$. Then, $\pottilde \in \mathcal{C}^2(\mathbb{R}^d)$ and
		\begin{equation}\label{eq:growth_zzp}
			\lim_{\lvert x\rvert \to \infty} \frac{\max(1,\lVert \nabla^2 \pottilde(x) \rVert)}{\lvert \nabla \pottilde(x) \rvert} = 0, \qquad  \lim_{\lvert x \rvert \to \infty} \frac{\lvert\nabla \pottilde(x)\rvert}{ \pottilde(x)} = 0.
		\end{equation}
		\item For $i=1,\dots,d$, the refreshment rates $\gammatilde_i:\R^d\to \R_+$ are continuous and such that $ 0 < \gammatilde_i(x) \leq  \overline{\gamma}<\infty$ for all $x\in\R^d$.
	\end{enumerate}
\end{assumption}
Under this assumption we have the following result. 
\begin{proposition}[Ergodicity of the time-changed ZZP]\label{thm:ergodicity_ZZP}
	Suppose \Cref{ass:s.integrability}, and \Cref{ass:ZZ_growth_proposal} hold. Then, the time-changed ZZP with speed $s$ and refreshment rates $\gamma_i= s \gammatilde_i$ for any $i=1\dots,d$ is \emph{geometrically ergodic}. Moreover, if $s(x) = \exp(\beta\pot(x))$ for $\beta\in(0,1)$, then $(X,V)$ is \emph{uniformly ergodic}.
\end{proposition}
\begin{proof}
	The proof is postponed to \Cref{sec:proof_ergo_ZZP}.
\end{proof}
\begin{remark}
	Notice that for the choice of speed function $s(x) = \exp(\beta\pot(x))$ we have $\pottilde =(1-\beta) \pot$. Therefore,  the growth conditions in \eqref{eq:growth_zzp} are satisfied if and only if the same conditions are satisfied with $\pot$ instead of $\pottilde$.
\end{remark}

Next we obtain a FCLT for the time-changed ZZP taking advantage of \Cref{thm:FCLT}. We shall focus on test-functions that depend only on the position vector, since this is typically the case of interest in the sampling literature.
\begin{proposition}[FCLT for the time-changed ZZP]\label{thm:FCLT_ZZP}
	Suppose \Cref{ass:s.integrability}, \Cref{ass:lln_base}, as well as \Cref{ass:ZZ_growth_proposal} hold. Let $(X,V)$ be the time-changed ZZP with speed $s$, and refreshment rates $\gamma_i= s \gammatilde_i$ for any $i=1\dots,d$. 
	Then, \eqref{eq:FCLT_convergence} and \eqref{eq:asymptotic_variance} hold for any $g:\R^d\to \R$ such that $\lvert g(x)\rvert \leq s(x) \exp(\beta\pot(x))$ for some $\beta\in(0,1)$.
\end{proposition}
\begin{proof}
	The result follows by the same arguments used in the proof of \Cref{thm:ergodicity_ZZP} (see \Cref{sec:proof_ergo_ZZP}). 
\end{proof}

\subsection{Convergence of time-changed Markov jump processes}\label{sec:conv.mjp}
We now specialise the main results of \Cref{sec:convergence_MP,sec:FCLT} to the jump processes described  as the third alternative of \Cref{sec:estimate_expectations}. These are based on applying a time-change with speed $s$ to a jump process $Y$ with rate one and kernel $\Qtilde$ which is $\mutilde$-stationary. 
Here we relate the convergence properties of the time-changed process $X$ to those of $\Qtilde$ and its associated Markov chain $(\Ybar_n)_{n\geq 0}$.
We work on the following assumption.
\begin{assumption}\label{ass:geoerg_discretetime}
	Consider a transition kernel $\Qtilde$ and an associated Markov chain $(\Ybar_n)_{n\geq 0}$. The following conditions are satisfied.
	\begin{enumerate}
        \item For all $f\in L^1(\mutilde)$, it holds that $N^{-1} \sum_{n=1}^{N} f(\Ybar_n) \to \mutilde(f)$ almost surely.
        \item There exists a bounded set $ C $ such that for some $n_*\in\N$, $c>0$, a non-trivial measure $\nu$, and a bounded set $  D \supset  C $, such that
		\[\mathbb{P}_y(\Ybar_{n_*}\in \cdot, \Ybar_i\in   D \textnormal{ for all } i=1,\dots,n_*) \geq c\nu(\cdot), \quad \textnormal{for all $y\in  C $. }\]
		\item There exists functions $\rmV:\rmE\to [1,\infty)$, with $\rmV \in \mathcal{A}$ and $\rmW:\rmE\to(0,1]$ and a constant $\eta\geq 0$ such that
		\begin{equation}\label{eq:drift_dicretetime}
            \Qtilde V(y) \leq \rmW(y) \rmV(y) + \eta\1_{ C }(y),
		\end{equation}
        where $ C $ is the same set of (2).
	\end{enumerate}
\end{assumption}
\Cref{ass:geoerg_discretetime}(2) is needed to ensure that the set $ C $ is petite for the jump process $X$ and can be verified in most cases of interest. \Cref{ass:geoerg_discretetime}(3) is instead used to obtain a drift condition for $X$. Notice that the case $\rmW(y) \leq \beta \in (0,1)$ corresponds to the drift condition ensuring geometric ergodicity of the Markov chain $(\Ybar_n)_{n\geq 0}$.

We are now ready to state our result on the geometric and uniform ergodicity of the process.
\begin{theorem}\label{prop:ergodicity_jumpproc}
	Suppose \Cref{ass:s.integrability} holds. Let $\Qtilde$ be a transition kernel that satisfies \Cref{ass:geoerg_discretetime} for some functions $\rmV \in \mathcal{A},\rmW$ and a set $ C $. 
    Then for the time-changed jump process with rate $s$ and kernel $\Qtilde$, the following hold.
    \begin{enumerate}
        \item the process is \emph{geometrically ergodic} with respect to $\mu$ if there exists a constant $\beta>0$ such that $s(z) \geq \nicefrac{\beta}{(1-\rmW(z))}$ for all $z\notin  C .$
        \item the process is \emph{uniformly ergodic} with respect to $\mu$ if there exists a constant $\beta>0$ such that $s(z) \geq \nicefrac{\beta \rmV(z)}{(1-\rmW(z))}$ for all $z\notin  C .$
    \end{enumerate}
\end{theorem}
\begin{proof}
   The proof is presented in \Cref{sec:proof_ergo_jumpproc}.
\end{proof}

We now give an analogue of \Cref{thm:FCLT}.
\begin{proposition}
    Suppose \Cref{ass:s.integrability} holds.
    Suppose $Y$ is a jump process with rate $1$ and jump kernel $\Qtilde$. 
    Suppose also that \Cref{ass:geoerg_discretetime} holds for some $\rmV \in \mathcal{A}, \rmW$ and a set $ C $. Assume that the speed function $\speed$ is such that $\speed(z)(1-\rmW(z))\rmV(z)\geq 1$ for all $z\in\rmE$.
    Finally, suppose $\mu(\rmV^2)<\infty$.
    Then, the FLCT result of \eqref{eq:FCLT_convergence} holds for any $g:\rmE\to\R$ such that $\lvert g(z)\rvert \leq \speed(z)(1-\rmW(z))\rmV(z)$, for $Z_n(t)$ as in \eqref{eq:Z_fctl}.
       
    Assuming also that $(1-\rmW(z))\rmV(z) \geq 1$ for all $z\in\rmE$, the asymptotic variance satisfies $\gamma^2_g = \gammatilde^2_{\tilde g},$ where $\tilde g(z) =\sqrt{\mu(s)} \,\, \frac{g(z)-\mu(g)}{s(z)}$ and $ \gammatilde^2_{f}$ is the asymptotic variance of the kernel $\Qtilde$.
\end{proposition}
\begin{proof}
	The proof is based on applying \Cref{thm:FCLT} and follows the same ideas of the proof of \Cref{prop:ergodicity_jumpproc}. The final result that $\gamma^2_g = \gammatilde^2_{\tilde g}$ is a simple consequence of the expression of the asymptotic variance of a kernel $\Qtilde$ obtained by \cite{glynn1996liapounov}.
\end{proof}

\section{Connection to space transformations}\label{sec:spacetransformations}

\begin{figure}[t]
\centering
	\begin{subfigure}[t]{0.47\textwidth}
		\includegraphics[width=\textwidth]{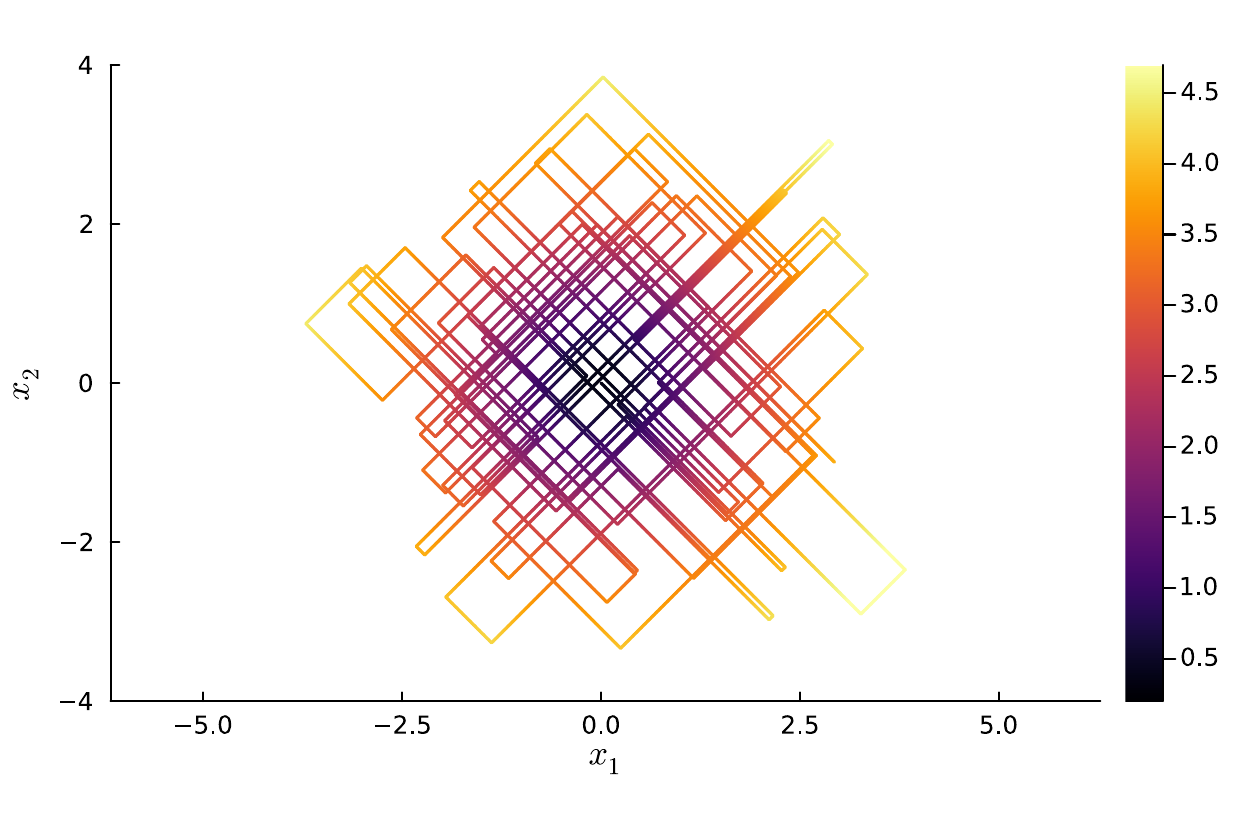}
        \vspace{-20pt}
		\caption{Time-changed ZZP with speed $s(x) = (1+\lvert x\rvert^2)^{(d+1)/2}$.}
        \label{fig:timechanged}
	\end{subfigure}
    \hfill
    \begin{subfigure}[t]{0.47\textwidth}
		\includegraphics[width=\textwidth]{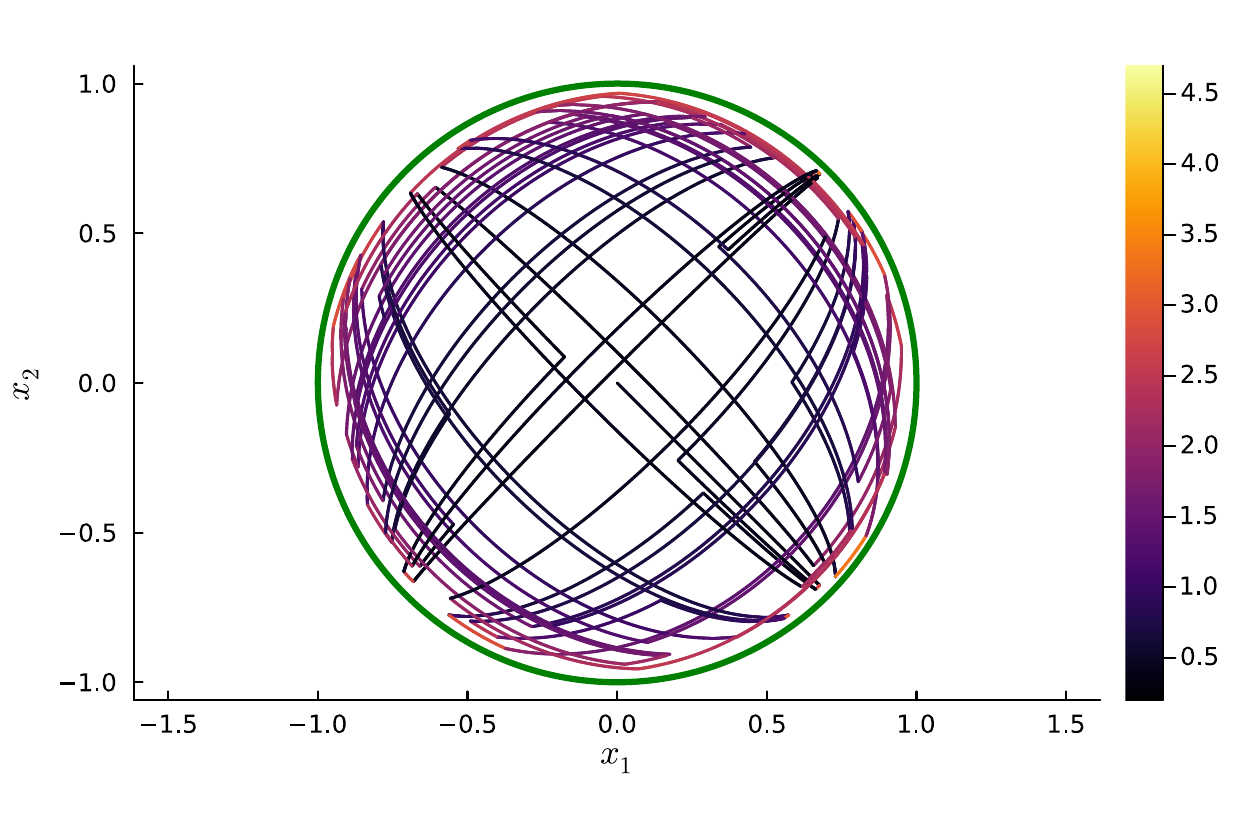}
        \vspace{-20pt}
		\caption{Space transform of the time-changed ZZP, with transformation $\rmH^{-1}(x) = x (1+\lvert x\rvert^2)^{-1/2}$.}
		\label{fig:timechanged_circle}
	\end{subfigure}
    \vspace{5pt}
    \begin{subfigure}[t]{0.47\textwidth}
		\includegraphics[width=\textwidth]{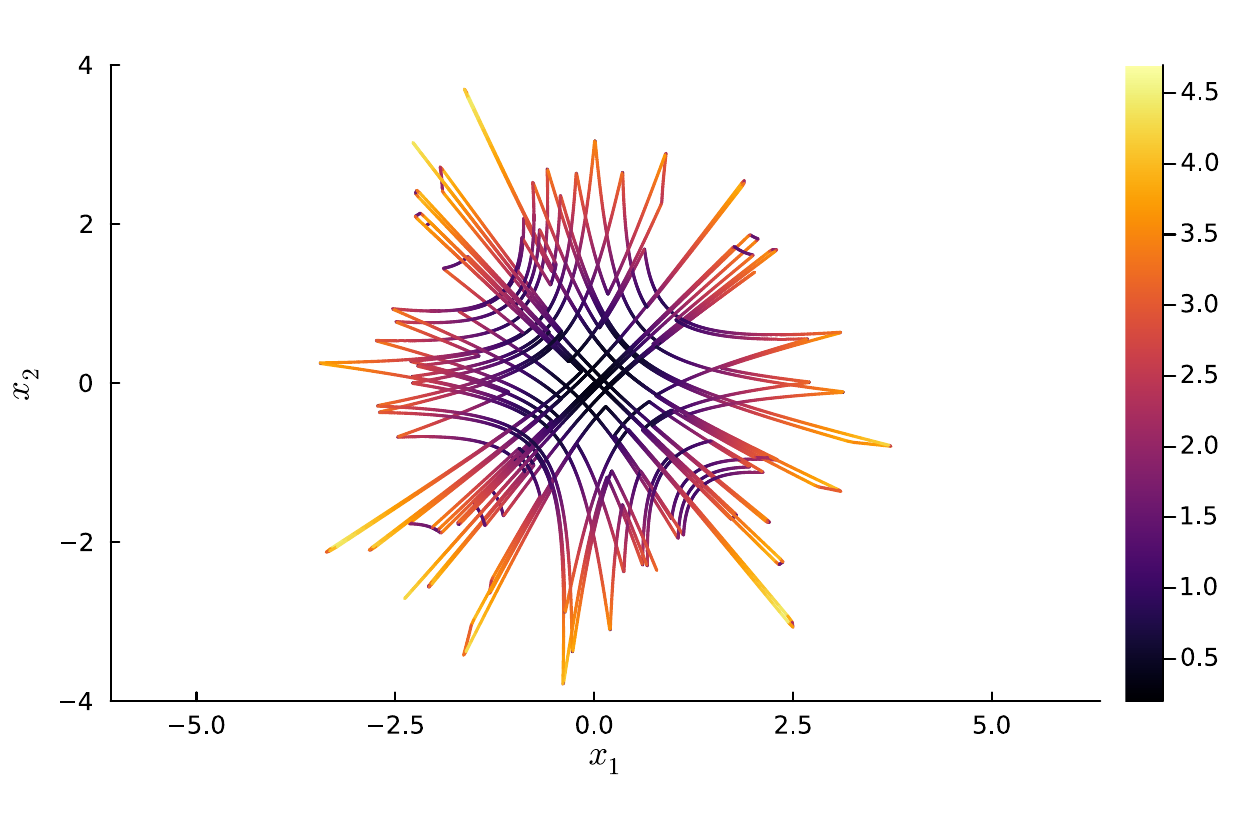}
        \vspace{-20pt}
		\caption{Space-transformed ZZP, with transformation $\rmH(y) = y (1-\lvert y\rvert^2)^{-1/2}$ for $y\in \{ z\in\R^2:\lvert z\rvert < 1\}.$}
		\label{fig:spacetransf}
	\end{subfigure}
    \hfill
	\begin{subfigure}[t]{0.47\textwidth}
		\includegraphics[width=\textwidth]{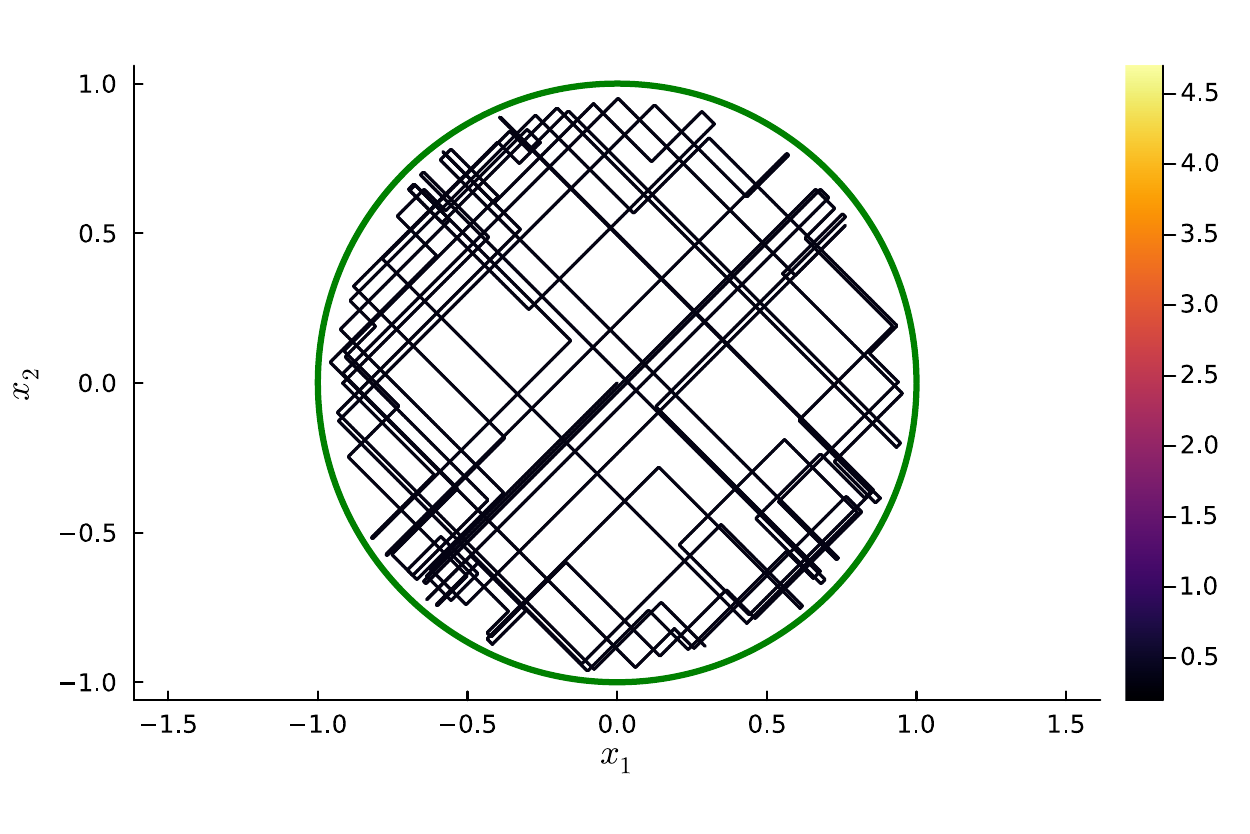}
        \vspace{-20pt}
		\caption{Base process for the space-transformed ZZP, that is a standard ZZP with target $\mutilde(y,w) = \mu(\rmH(y),w) s(\rmH(y))$ with $y\in \{ z\in\R^2:\lvert z\rvert < 1\}.$}
		\label{fig:spacetransf_circle}
	\end{subfigure}
	\caption{Comparison between time-change and space transformation, using the standard ZZP as base process and a two-dimensional, standard normal target distribution. The choices of the speed function and diffeomorphism are as in \eqref{eq:speed_figure} and \eqref{eq:diffeomorphism_figure}. 
    The colours and the scale on the right of each plot represent the logarithm of the speed of each process, calculated as distance travelled in a small time unit divided by the length of the time unit.}
	\label{fig:space_vs_time_transf}
\end{figure}

\cite{Johnson_Geyer} suggested another approach to improve convergence of MCMC algorithms in the context of heavy tailed target distributions relying on a diffeomorphism $\rmH:\R^d\to\R^d$. Their idea is to apply the change of variables $x = \rmH(y)$ to transform the original target density $\mu$ to $\mutilde(y) = \mu(\rmH(y)) \left\lvert \det \rmJ_{\rmH}(y) \right\rvert$ where $\rmJ_{\rmH}$ denotes the Jacobian matrix of $\rmH$. 
Choosing $\rmH$ suitably gives that $\mutilde$ has lighter tails than $\mu$ and thus it can be easier to target using a Markov process $Y$ that is stationary with respect to $\mutilde$. Then, one can apply the map $\rmH$ to define the process $X_t = \rmH(Y_t)$, which can be used to obtain asymptotically exact samples from $\mu$. 
It is straightforward to check that the process $X$ has stationary distribution $\mu$ and inherits the convergence properties of $Y$, hence enabling e.g. geometric convergence even when $\mu$ is heavy tailed. 

The approaches of space and time transformations are strongly related. In order to see this, it is helpful to rewrite $\mutilde$ as
\begin{equation}\label{eq:target_spacetransf}
	\mutilde(y) = \mu(x) \left\lvert \det \rmJ_{\rmH}(\rmH^{-1}(x))\right\rvert, \quad \text{for } x=\rmH(y).
\end{equation}
This shows that both strategies are built on a base process that targets a distribution $\mutilde$ of the form $\mutilde \propto s \mu$, where $s$ represents either the speed function that defines the time change or the Jacobian determinant term of the diffeomorphism $\rmH$. It is indeed reasonable to interpret the Jacobian determinant as a speed function, since it measures how infinitesimal volumes are affected by the map $\rmH$ and hence how fast the process $X$ moves through a unit of volume.

Following this connection, we can obtain novel speed functions as Jacobian determinants of a suitable space transformation.
For instance, consider the diffeomorphism defined for points $y \in \{z \in\R^d: \lvert z\rvert < 1 \}$ as
\begin{equation}\label{eq:diffeomorphism_figure}
    \rmH(y) = \frac{y}{(1-\lvert y\rvert^2)^{1/2}}\,.
\end{equation}
This has inverse $\rmH^{-1}(x) = x (1+\lvert x\rvert^2)^{-1/2},$ which maps points $x\in\R^d$ to points $y = \rmH^{-1}(x)$ inside the $d$-dimensional unit sphere. It is not difficult to obtain that the associated Jacobian determinant satisfies $\det \rmJ_{\rmH}(y)= (1-\lvert y\rvert^2)^{-(d+1)/2}$, which motivates using as speed function $s(x) = \lvert \det \rmJ_{\rmH}(\rmH^{-1}(x))\rvert$, i.e.
\begin{equation}\label{eq:speed_figure}
    s(x) = (1+\lvert x\rvert^2)^{(d+1)/2}\,.
\end{equation}
Intuitively, this transformation is helpful when one wants to explore the tails of a distribution that is centred around zero. In the case of time-changed ZZP or the time-changed BPS, using a speed function of the form \eqref{eq:speed_figure} has the additional advantage that the deterministic dynamics can be simulated analytically, without the need of numerical integration \citep{Vasdekis_speedup}. In \Cref{fig:timechanged} and \Cref{fig:spacetransf} we compare the approaches of time and space transformations obtained respectively with speed \eqref{eq:speed_figure} and diffeomorphism \eqref{eq:diffeomorphism_figure}, using the ZZP as base process in each case. The time-changed ZZP, shown in \Cref{fig:timechanged}, moves in straight lines with speed that increases with the distance from the origin, regardless of the direction of the process. \Cref{fig:timechanged_circle} shows the trace plot of the 
process obtained applying the transformation $\rmH^{-1}$ to the position vector of the time-changed ZZP. This process has constant speed when it moves radially, while it speeds up when it moves with a tangential component. Because such base process is compactly supported and has speed that is lower bounded by a positive constant (see \Cref{fig:timechanged_circle}), it is then reasonable to expect that it can be uniformly ergodic, and as a consequence so is the time-changed ZZP \citep{Johnson_Geyer}.
On the other hand, the space-transformed ZZP, shown in \Cref{fig:spacetransf}, moves in curved lines and speeds up when moving radially away from the origin, while it slows down when moving with tangential component to the level curves of the Jacobian determinant.

Similarly, we can obtain other speed functions that are the Jacobian determinant of diffeomorphisms from (a compact subset of) $\R^d$ to $\R^d$.
A related example is the stereographic map, recently considered by  \citet{yang2024stereographicmarkovchainmonte}, that transforms points in $\R^d$ to points on the surface of the unit sphere in $\R^{d+1}$ (excluding the North pole). The Jacobian determinant of the stereographic projection turns out to be (up to a constant factor) $s(x)=(1+\lvert x\rvert^2 )^d,$ that is similar to \eqref{eq:speed_figure}.

\section{Numerical simulations}\label{sec:numerical_simulations}
In this section we test algorithms based on the time-change approach on two toy examples: a multimodal setting and a rare event estimation task for a heavy tailed distribution. In both cases, we construct a jump process with jump rate $s$ and jump kernel that is the one-step transition kernel of the Metropolis-adjusted ZZP, a discrete-time chain recently introduced in \citet{bertazzi_splitting}. We then use this process as prescribed by Equations \eqref{eq:erg_avg_jumpprocess} and \eqref{eq:erg_avg_jumpprocess_discrete}.
The codes for all the experiments in this paper can be found at \url{https://github.com/andreabertazzi/time_changed_sampling}.

\subsection{Sampling from a mixture of normal distributions}\label{ex:gaussianmixture}

\begin{figure}[t]
	\begin{subfigure}[t]{0.4\textwidth}
		\centering
		\includegraphics[width=\textwidth]{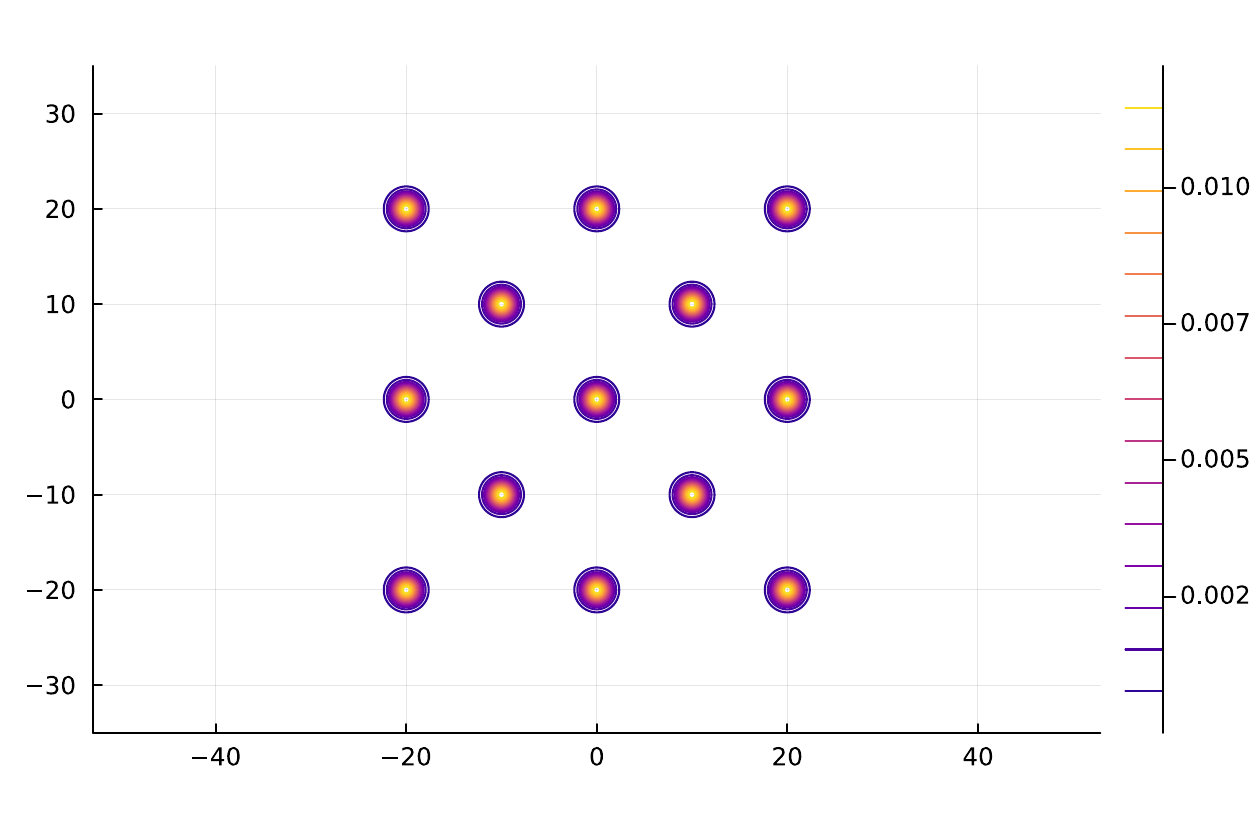}
        \vspace{-20pt}
		\caption{Level curves of $\mu$.}
		\label{fig:mixture_target}
	\end{subfigure}

	\begin{subfigure}[t]{0.4\textwidth}
		\includegraphics[width=\textwidth]{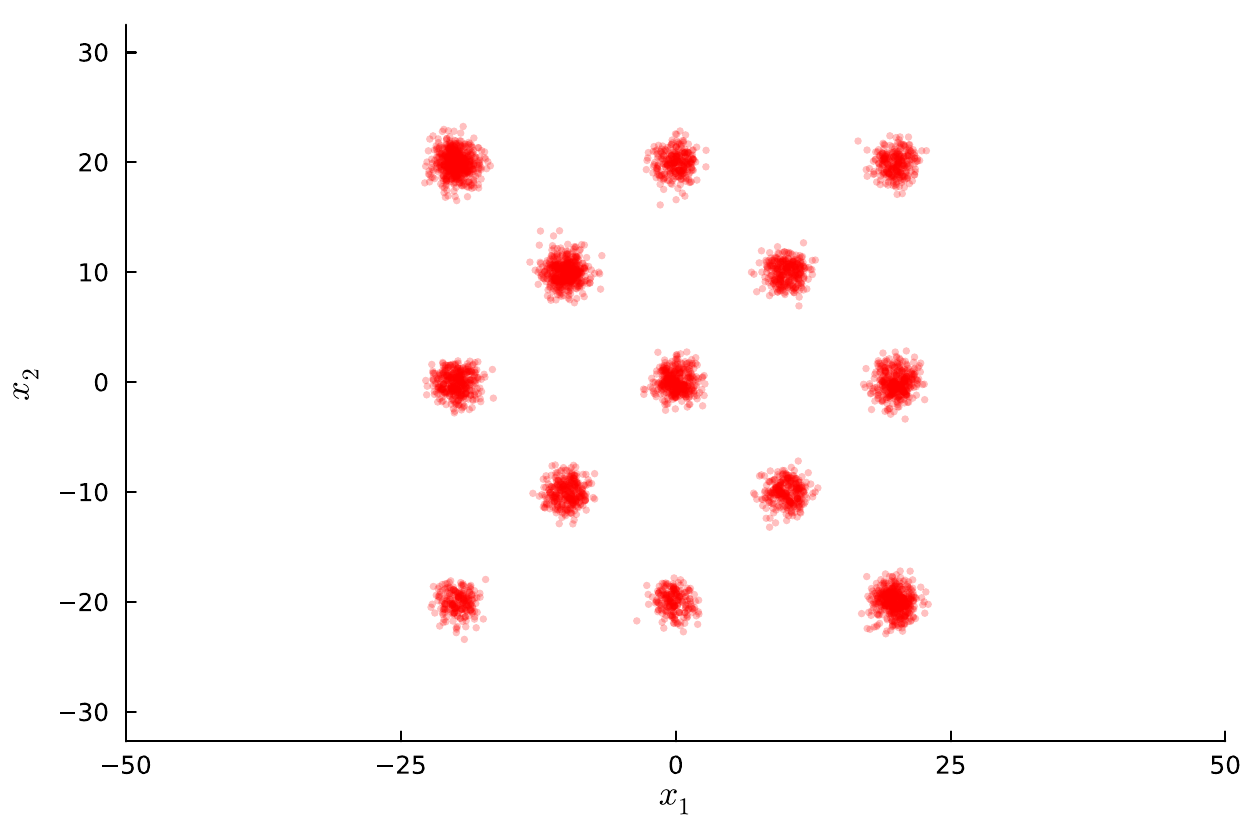}
		\caption{Discretisation of the time-changed jump process with $s(x)=\mu(x)^{-0.9}$.}
		\label{fig:Xtilde}
	\end{subfigure}
	\hspace{10pt}
	\begin{subfigure}[t]{0.4\textwidth}
		\includegraphics[width=\textwidth]{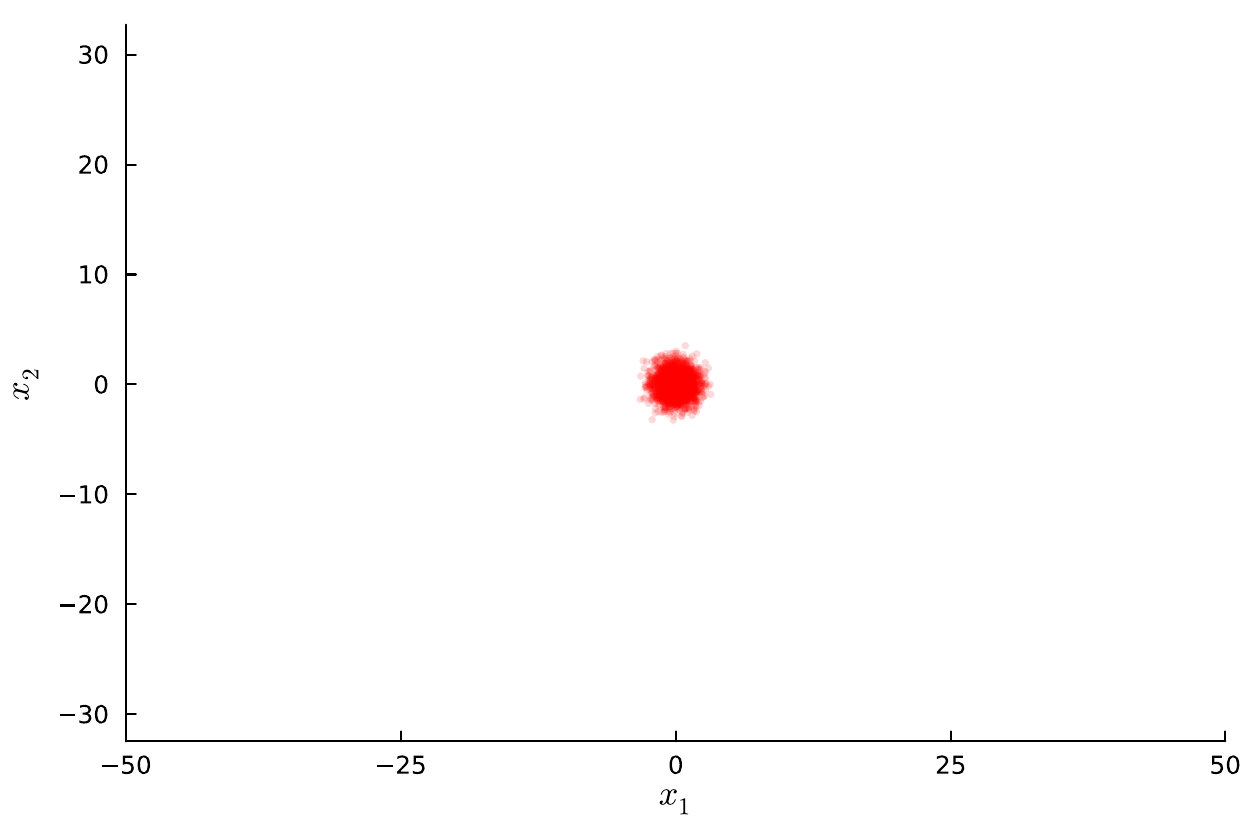}
		\caption{Discretisation without time-change ($s(x)=1$).}
		\label{fig:s=1}
	\end{subfigure}
	\caption{Numerical simulations in the context of Section \ref{ex:gaussianmixture}. 
    The plots in the second row are obtained as follows. First, we simulate a jump process with initial condition at the origin, jump rate $s$, and jump kernel given by the Metropolis adjusted ZZP \citep{bertazzi_splitting} with random step size $\Exp(1/\delta)$ for $\delta = 0.1$ and invariant distribution $\mutilde \propto s\mu$.
    Then, we discretise the obtained paths with step size $10^{-2}$. In either case, we performed $5\times 10^4$ iterations of the Metropolis adjusted ZZP.}
	\label{fig:gaussianmixture}
\end{figure}
We consider the case where the target $\mu$ is a mixture of thirteen, two-dimensional standard Gaussian distributions with equal weights as shown in Figure~\ref{fig:mixture_target}. 
In order to sample from this distribution, we simulate the jump  process based on the Metropolis-adjusted ZZP and then perform an a posteriori discretisation of the path as prescribed in \eqref{eq:erg_avg_jumpprocess_discrete}. This approach discards samples that are far from the modes with high probability.
Figures~\ref{fig:Xtilde} and \ref{fig:s=1} show the trace plots obtained using the speed functions $s(x)=\mu(x)^{-0.9}$ and $s(x) = 1$ respectively. In the former case, the process successfully visits all modes, while the latter, standard sampler remains stuck in the mode it is initialised in.

\subsection{Sampling from a heavy-tailed distribution}\label{sec:heavytailed}
\begin{figure}[t]
	\begin{subfigure}[t]{0.45\textwidth}
		\includegraphics[width=\textwidth]{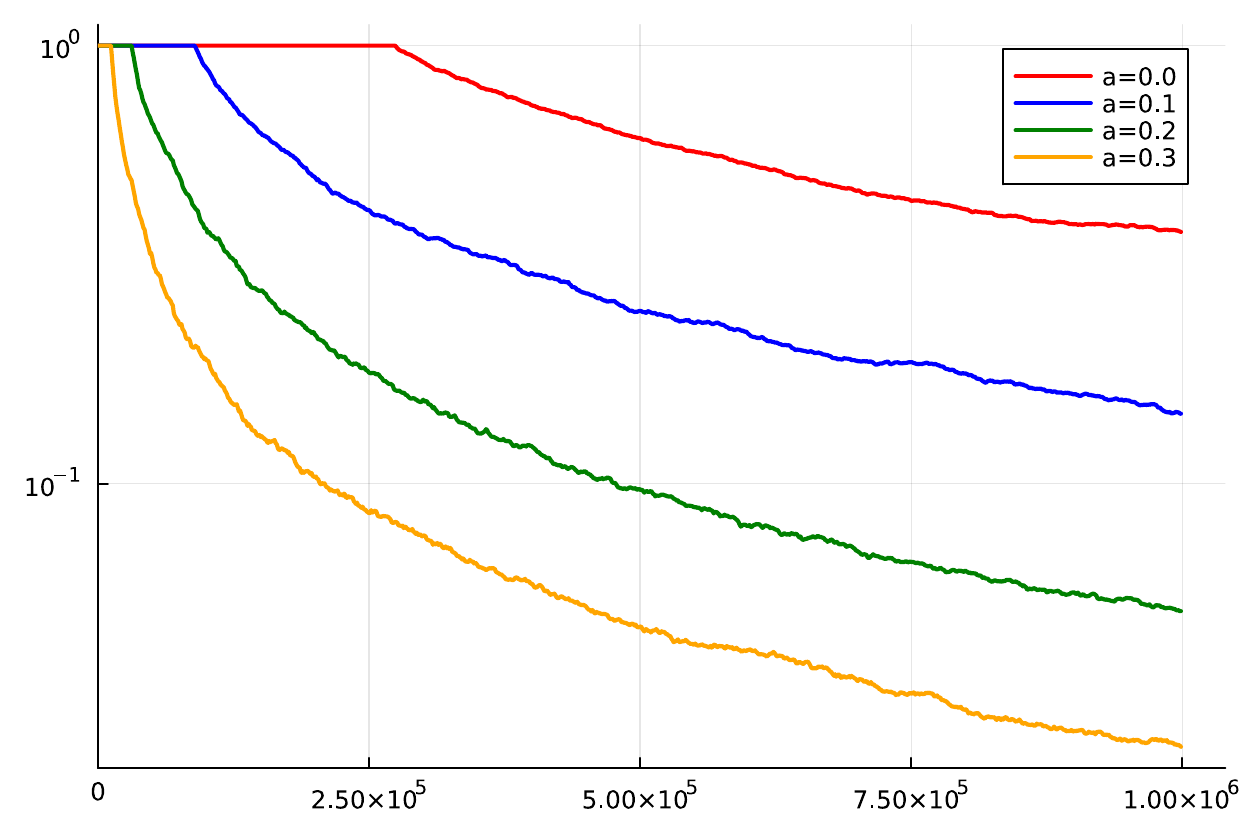}
		\caption{Median of the relative square error ($y$-axis) as a function of the number of jumps of the base process ($x$-axis).}
        \label{fig:heavytailed_mse}
	\end{subfigure}
\hfill
	\begin{subfigure}[t]{0.45\textwidth}
		\includegraphics[width=\textwidth]{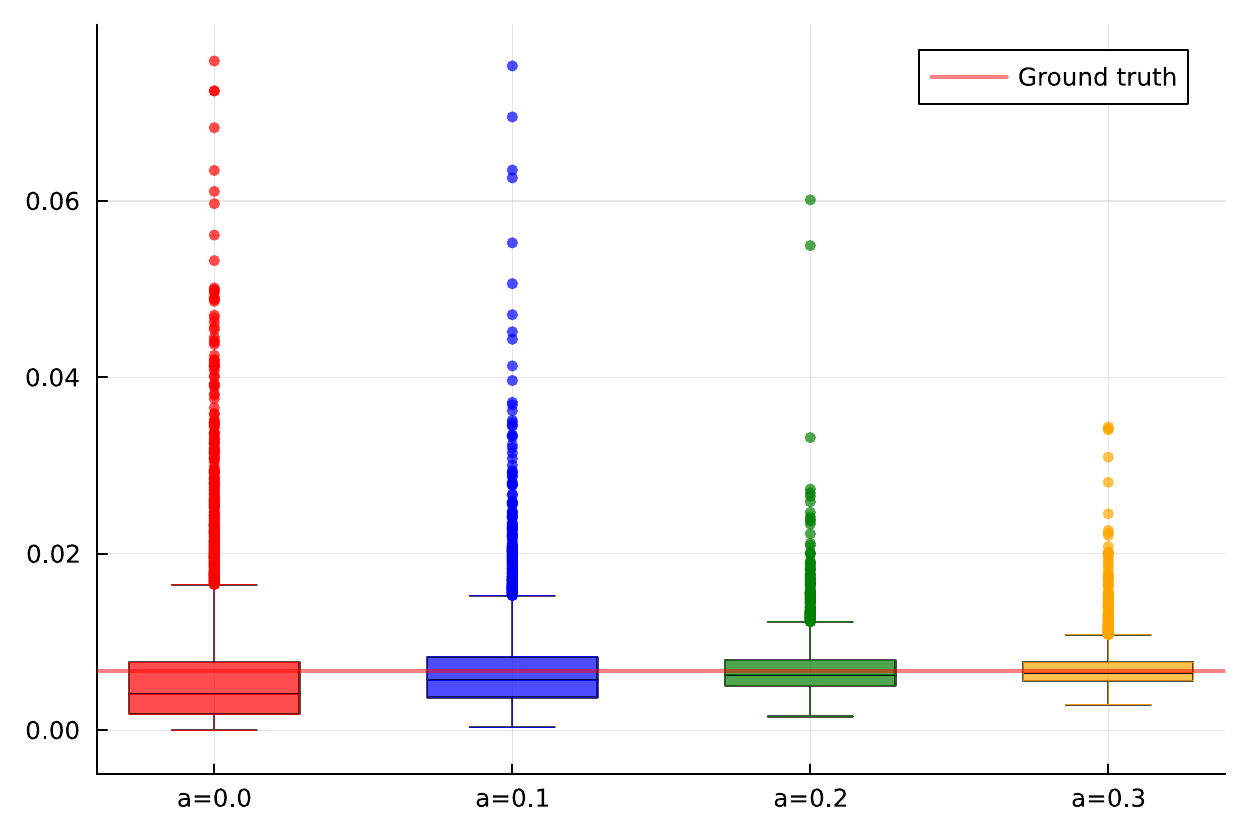}
		\caption{Estimates for the probability of the set $\{x\in\R^2: \lvert x\rvert > 150\}$ for different values of $a$.}
	\end{subfigure}
	\caption{Numerical simulations in the context of \Cref{sec:heavytailed}. The plots show estimates of the probability of the event $\{x\in\R^2: \lvert x\rvert > 150\}$ obtained simulating $5000$ independent runs of the jump process based on the Metropolis-adjusted ZZP for $10^6$ iterations.}
	\label{fig:heavytailed}
\end{figure}

Consider a $2$-dimensional t-distribution with one degree of freedom and identity covariance matrix. This distribution has density proportional to $(1+\lvert x\rvert^2)^{-3/2}$ for $x\in\R^2$.
We consider the task of estimating the probability assigned to the set $\{x\in\R^2: \lvert x\rvert > 150\}$. Since we want to encourage the process to visit the tails of the target distribution, we take speed functions of the form $s(x)=\exp(a\pot(x))$, where $0 \leq a<1/3$ ensures that $\mu(s)<\infty$. In this case we follow the approach \eqref{eq:erg_avg_jumpprocess} since samples in the tails are important to obtain a good estimator. The results are shown in Figure~\ref{fig:heavytailed}. Figure~\ref{fig:heavytailed_mse} shows that the introduction of the speed function decreases the mean squared error of the estimator. For $a=0.3$ we obtain the estimator with the smallest MSE.

\section{Discussion}\label{sec:discussion}

In this paper we have discussed the use of time-changes of Markov processes in the context of MCMC algorithms. As we have seen, this concept is closely related to performing importance sampling on the paths of a base process with biased stationary distribution.
We discussed suitable speed functions that can be beneficial in the context of multimodal and heavy tailed distributions and illustrated their strengths with simple numerical simulations.
A conceptual comparison to the framework of space transformations \citep{Johnson_Geyer} highlighted that time-changed processes can be more intuitive to use for tail exploration (see \Cref{sec:spacetransformations}). Similarly, it appears simpler to find a good speed function in the context of multimodal distributions compared to a good diffeomorphism that enhances mode exploration. On the other hand, space transformations seem preferable when the target has complex geometry \citep{Girolami_Calderhead_riemannHMC}.

There are several directions which we leave as topics for future research. For example, although we gave intuitive choices of speed function, these tend to depend on a parameter (e.g. $s(x) =\mu^{-a}(x)$ for $a\in(0,1)$). It is then natural to wonder how to learn this parameter, or also how to let it dynamically evolve in order to alternate between exploration and exploitation phases. 
At the same time, another application we have not explored is the use of time-changes to obtain more accurate approximations of stochastic processes. Indeed, it can be useful to take advantage of a time-changed process to slow down a process to incur in smaller numerical errors.

\section*{Acknowledgements}
The authors would like to thank Joris Bierkens, Paul Dobson, Pierre Monmarché, Samuel Power, and Samuel Livingstone for helpful discussions and comments on various versions of this manuscript. 

AB is funded by the European Union (ERC-2022-SyG, 101071601). Views and opinions expressed are however those of the authors only and do not necessarily reflect those of the European Union or the European Research Council Executive Agency. Neither the European Union nor the granting authority can be held responsible for them.

The authors would like to thank the Isaac Newton Institute for Mathematical Sciences, Cambridge, for support and hospitality during the programme Stochastic systems for anomalous diffusion where work on this paper was undertaken. This work was supported by EPSRC (EP/Z000580/1). The authors would also like to thank Codina Cotar for inviting them to participate at the INI programme.

\bibliographystyle{plainnat}
\bibliography{main.bib}

\appendix 

\section{Proof of \Cref{thm:lln_timechange}}\label{sec:proof_LLN}
We first prove the law of large numbers  \eqref{lln.transform:1}. Let $x \in \rmE$ and $f \in L^1(\mu)$. We observe that by the change of variables $u=r(t)$ we get
    \begin{align}
        \frac{1}{T} \int_0^T f(X_t) \dd t &= \frac{1}{T}\int_0^Tf(Y_{r(t)}) \dd t \\
        & = \frac{1}{T} \int_0^{r(T)} f(Y_u) \frac{1}{s(Y_u)} \dd u \\
        & =\frac{r(T)}{T} \frac{1}{r(T)} \int_0^{r(T)} f(Y_u) \frac{1}{s(Y_u)} \dd u.\label{lln.proof.eq.1}
    \end{align}
    Therefore, due to \Cref{ass:lln_base}, we find
     \begin{equation}\label{lln.proof.eq.2}
         \frac{T}{r(T)} \, \frac{1}{T} \int_0^T f(X_t) \dd t = \frac{1}{r(T)} \int_0^{r(T)} f(Y_u)\frac{1}{s(Y_u)} \dd u \xrightarrow[a.s.]{T \rightarrow \infty} \frac{1}{\mu(s)} \int f(y) \mu(\dd y).
     \end{equation}
     This holds for any $f \in L^1(\mu)$. In particular, it holds when $f \equiv 1$, in which case \eqref{lln.proof.eq.2} gives
     \begin{equation}
         \frac{r(T)}{T} \xrightarrow[a.s.]{T \rightarrow \infty} \mu(s).
     \end{equation}
     Therefore, from \eqref{lln.proof.eq.1} we get
     \begin{equation}\label{eq:lln_proof}
         \frac{1}{T} \int_0^T f(X_t) \dd t \, \xrightarrow[a.s.]{ T \rightarrow \infty} \, \int f(y) \mu(\dd y) ,
     \end{equation}
    as required. 
    
    Now, let $A \subset E$ be a borel set with $\mu(A)>0$. Using the law of large numbers we get that
    \begin{equation}
    \mathbb{P}\left( \int_0^{\infty} 1_{X_t \in A}=\infty \right)=1,
    \end{equation}
    meaning that $X$ is Harris recurrent in the sense of \citet{MeynTweedie_2}, and therefore it possesses a unique invariant measure $\nu$ (see e.g. \cite{Azema1969, MeynTweedie_2}). From  \eqref{eq:lln_proof} we get that $\nu=\mu$ and this concludes the proof.

\section{General results for  time-changes}\label{appendix:theory}

\subsection{Weak Generator}\label{appendix:generator}
We turn our attention to the \emph{weak generator} of the time-changed process, a mathematical object that is essential in the analysis of Markov processes. The weak generator was defined in \Cref{sec:convergence_MP}.

We shall assume throughout that the martingale problem for $Y_t$ is well posed, that is it admits a unique solution. The following theorem shows that for a large class of functions, the generator of $(X_t)_{t\geq 0}$ is obtained by multiplying the generator of $(Y_t)_{t\geq 0}$ by the speed function $s$. Heuristically, this is to be expected since the generator represents the derivative of a test function with respect to the dynamics of the process, so speeding up the time by a function $s$ should be reflected in the derivative. We have the following.
\begin{proposition}\label{thm:generator_timechange}
	Suppose that the Markov process $(Y_t)_{t\geq 0}$ has a weak generator $(\cLtilde,\cD(\cLtilde))$. Let $\mathcal{C} \subset \cD(\cLtilde) \cap C_b(\rmE)$ be such that for any $f \in \mathcal{C}$, $s \cLtilde f$ is a bounded function. Then, $\mathcal{C}$ is in the domain of the weak generator of the time-changed process $(X_t)_{t\geq 0}$, defined via \eqref{eq:time_change}, Furthermore, for the weak generator $\cL$ of the process $X$, we have that for all $f \in \mathcal{C}$,
 \begin{equation}
     \cL f= s \cLtilde f,
 \end{equation} 
 meaning that $\mathcal{C} \subset \mathcal{A}$, where $\mathcal{A}$ is defined in \Cref{defn:nice.set}.
\end{proposition}
\begin{proof}
	 The proof is obtained using the same argument as in the proof of \cite[Theorem 1.3, Chapter 6]{Ethier}.
\end{proof}

\Cref{thm:generator_timechange} partially justifies a technical assumption we make throughout our main results in \Cref{sec:theory}, i.e. that the Lyapunov function $\rmV \in \mathcal{A}$. The proposition is also used in \Cref{proof.diffusion.time.change} to study time-changes of diffusions. In most practical cases, the set $\mathcal{C}$, and therefore the set $\mathcal{A}$ as well, will contain a large family of functions. For example, for our running example, the Zig-Zag process, the infinitely differentiable functions with compact support are a subset of $\mathcal{C}$. The same holds for other PDMPs that move around the space with constant speed, such as the Bouncy Particle Sampler with velocity refreshments chosen from the unit sphere, but also for diffusion processes as well. 

\subsection{Irreducibility}\label{sec:tt_properties_irred}
Let us consider the concept of \emph{irreducibility}, which is important to obtain convergence of the law of the process to its stationary distribution. The process $Y_t$ is $\psi$-irreducible if for some non-trivial measure $\psi$, for any measurable set $\rmB$ such that $\psi(\rmB)>0$ it holds $\mathbb{E}_y[\int_0^\infty \1_{Y_t\in \rmB}\, \dd t]>0$. In the next proposition we show that irreducibility of $X_t$ follows from irreducibility of $Y_t.$
\begin{proposition}\label{prop:irreducibility}
	Let $Y_t$ and $X_t$ on $\rmE$ be two Markov processes related by \eqref{eq:time_change}, where $s$ is bounded on bounded sets. 
	Let $\psi$ be a measure that 
    is absolutely continuous with respect to Lebesgue measure. 
	Then, $X_t$ is $\psi$-irreducible if and only if $Y_t$ is $\psi$-irreducible. 
\end{proposition}
\begin{proof}
	Assume $Y_t$ is $\psi$-irreducible and recall the relation \eqref{eq:time_transformed_process_inverse}. Then for any $\rmB$ such that $\psi(\rmB)>0$ we have by a change of variables
	\begin{align}
		\mathbb{E}_y\left[\int_0^\infty \1_{X_t\in \rmB} \dd t\right] = \mathbb{E}_y\left[\int_0^\infty \1_{Y_t\in \rmB}\,\frac{1}{s(Y_t)} \dd t\right].
	\end{align}
	If $\rmB$ is a bounded set, then there exists $\overline s_{\rmB}<\infty$ such that $s(x) \leq s_{\rmB}$ for all $x\in\rmB$. Combined with the fact that $\mathbb{E}_y\left[\int_0^\infty \1_{Y_t\in \rmB} \dd t\right]>0$, from the irreducibility of $Y$, we get that $\mathbb{E}_y\left[\int_0^\infty \1_{X_t\in \rmB} \dd t\right]>0.$ Suppose now $\rmB$ is not bounded but has non-zero Lebesgue measure. In this case there exists $\rmB_1 \subset \rmB $  that is bounded and has positive Lebesgue measure. Then $\mathbb{E}_y\left[\int_0^\infty \1_{X_t\in \rmB} \dd t\right]\geq \mathbb{E}_y\left[\int_0^\infty \1_{X_t\in \rmB_1} \dd t\right]$ which is strictly positive as shown above.
    Therefore, $X_t$ is $\psi$-irreducible.

    The reverse statement can be obtained observing that, since $s(x)\geq \underline{s}>0$ for all $x\in\rmE$,
    \begin{align}
		\mathbb{E}_y\left[\int_0^\infty \1_{Y_t\in \rmB} \,\dd t\right] &= \mathbb{E}_y\left[\int_0^\infty \1_{X_t\in \rmB}\,s(X_t) \,\dd t\right]  \\
        & \geq \underline{s} \, \mathbb{E}_y\left[\int_0^\infty \1_{X_t\in \rmB}\, \dd t\right]
	\end{align}
    which is positive since $X_t$ is irreducible.
\end{proof}

\subsection{Petite sets}
Let us now introduce the important notion of \emph{small set}, which is often used to establish ergodicity of a Markov processes \cite{DownMeynTweedie} and MCMC algorithms \cite{rob_ros_mcmc_survey}. A set $ C  \subset \rmE$ is $(t_0,\eta,\nu)$-small for a process $Y_t$ if there exist a constant $\eta>0$ and a non-trivial probability measure $\nu$ on $\rmE$, i.e. such that $\nu(\rmE)>0$, such that for all $y\in  C $  $$\mathbb{P}_y(Y_{t_0}\in\cdot\,) \geq \eta \nu(\cdot).$$ 
A related, weaker condition is that of \emph{petite set}: there exists a probability measure $\alpha$ on $(0,\infty)$ such that for all $y\in  C $
$$\int_0^\infty \mathbb{P}_y(Y_t\in \cdot\,) \,\alpha(\dd t)  \geq \eta \nu(\cdot).$$ 
We are interested in proving that under suitable conditions a small or petite set for $Y_t$ can be petite for the time changed process $X_t$. For this to hold, we are imposing \Cref{ass:generic_drift_condition_tt}(1) from \Cref{sec:convergence_MP}.

The goal of this section is to show that \Cref{ass:generic_drift_condition_tt}(1) is sufficient to conclude that $ C $ is petite for the time-changed process $X$. The essential idea is that on the event in which $Y_t$ does not leave the bounded set $  C $ it is possible to bound the time-change $r(t)$. Notice that \Cref{ass:generic_drift_condition_tt}(1) holds e.g. when $Y$ satisfies a uniform small set condition of the form $\mathbb{P}_y(Y_{t}\in\cdot\,) \geq \eta \nu(\cdot)$ for all $t \in [t_0,t_0+\varepsilon]$ and $y \in C$, and also $Y_t\in  D$ almost surely for $t\leq t_0+\varepsilon$ whenever $Y_0\in C $. This will be the case for any PDMP with deterministic dynamics having a bounded velocity, such as the Zig-Zag sampler, Bouncy particle sampler, and the Boomerang sampler.

\begin{lemma}\label{lem:small_sets}
	Suppose $Y$ satisfies \Cref{ass:generic_drift_condition_tt}(1) for some sets $ C  \subset D$, and moreover that \Cref{ass:s.integrability} holds. 
	Then, $ C $ is petite for the time-changed process $X$.
\end{lemma}
\begin{proof}
    Let $\beta$ be the uniform distribution on the interval $[t_1,t_2]$ for some $t_1<t_2$ that will be defined later, and denote $\Delta t = t_2-t_1$.
    For any $x\in  C $ and for any measurable set $ A $ it holds that
	\begin{align}
        \int \mathbb{P}_x(X_{t}\in  A ) \beta(t) \dd t & = \frac1{\Delta t}\, \int_{t_1}^{t_2} \mathbb{P}_x(X_{t}\in  A ) \, \dd t \\
		& =  \frac1{\Delta t}\, \mathbb{E}_x \left[  \int_{t_1}^{t_2}\1_{Y_{r(t)}\in  A } \dd t \right].
	\end{align}
	In the last line we applied Fubini's theorem. Applying the change of variables $t'=r(t)$ we find
	\begin{align}
		\int\mathbb{P}_x(X_{t}\in  A ) \beta(t) \dd t & = \frac1{\Delta t} \mathbb{E}_x \left[  \int_{r(t_1)}^{r(t_2)}\frac{1}{s(Y_t)}  \, \1_{Y_{t}\in  A }  \dd t \right]  \\
        &\geq \frac1{\overline{s}_{  C }\Delta t} \mathbb{E}_x \left[  \int_{r(t_1)}^{r(t_2)} \, \1_{Y_{t}\in  A ,\, Y_u\in   D  \textnormal{ for all } u\in [0,t]}\,  \dd t \right] .
	\end{align}
    Now notice that $r(t_2) \geq \underline{s} t_2$ almost surely. Similarly, on the event that $Y_u\in   D \textnormal{ for all } u\in [0,t]$ for $t\geq r(t_1)$, we have $r(t_1) \leq \overline{s}_{  C } t_1$ almost surely. 
	Therefore we find 
	\begin{align}
		\int\mathbb{P}_x(X_{t}\in  A ) \beta(t) \dd t & \geq \frac1{\overline{s}_{  C }\Delta t} \int_{\overline{s}_{  C } t_1}^{\underline{s} t_2} \mathbb{P}_x ( Y_{t} \in  A , Y_u\in   D  \textnormal{ for all } u\in [0,t]  ) \dd t
	\end{align}
    since we are effectively restricting the domain of integration. Choosing $t_2 = \frac{t_0+\varepsilon}{\underline{s}}$ and $t_1 = \frac{t_0}{\overline{s}_{  C }}$ gives
	\begin{align}
		\int\mathbb{P}_x(X_{t}\in  A ) \beta(t) \dd t & \geq \frac1{\overline{s}_{  C }\Delta t} \int_{t_0}^{t_0 + \varepsilon} \mathbb{P}_x ( Y_{t} \in  A , Y_u\in   D  \textnormal{ for all } u\in [0,t]  ) \dd t\\
        & \geq \frac1{\overline{s}_{  C }\Delta t} \eta \nu(A),
	\end{align}    
    which concludes the proof.
\end{proof}

\section{Proofs of \Cref{sec:examples}}
\subsection{Characterisation of time-changed PDMPs}\label{sec:proof_timechange_pdmp}
Let us recall the construction of a time-homogeneous PDMP $(Y,W)$ with characteristics $(\vftilde, \lambdatilde, \Qtilde)$. Let $(E_m)_{m=1,2,\dots}$ and $(U_m)_{m=1,2,\dots}$ denote two independent sequences of i.i.d. random variables respectively with the exponential distribution with parameter $1$ and the uniform distribution on $[0,1]$. Denote as $\varphitilde_t$ the flow map corresponding to the ODE $\dd (y_t,w_t) = \vftilde(y_t,w_t) \dd t$. We introduce the function $\rmF:\rmE \times [0,1] \mapsto \rmE$ such that for $U\sim\Unif([0,1])$ it holds that $\mathbb{P}(\rmF((y,w),U) \in A) = \Qtilde((y,w),A)$ for all measurable sets $A$ and all $(y,w)\in\rmE.$
We inductively construct the event times $(\tilde T_m)_{m=0,1,2,\dots}$ We first set $\tilde T_0=0$. Assume that $\tilde T_1,\dots,\tilde T_n$ and the realisation of the path $(Y_t,W_t)_{t\leq T_n}$ are known. We shall show how $\tilde T_{n+1}$ and $(Y_t,W_t)_{\tilde T_n<t\leq \tilde T_{n+1}}$ can be obtained. First of all, $\tilde T_{n+1} = \tilde T_n + \tilde \tau_{n+1}$ where the interarrival time $\tilde \tau_{n+1}$ is given by
\begin{equation}\label{eq:interarrival_pdmp}
    \tilde \tau_{n+1} = \inf\left\{t>0: \int_0^t \lambdatilde(\varphitilde_u(Y_{\tilde T_n},W_{\tilde T_n})) \dd u \geq E_{n+1} \right\}.
\end{equation}
Then one lets $(Y_{\tilde T_n + t},W_{\tilde T_n + t}) = \varphitilde_t(Y_{\tilde T_n},W_{\tilde T_n})$ for $0\leq t< \tilde \tau_{n+1}$ and 
$$(Y_{\tilde T_{n+1}},W_{\tilde T_{n+1}}) = \rmF(\varphitilde_{\tau_{n+1}}(Y_{\tilde T_n},Y_{\tilde T_n}),U_{n+1}).$$
The whole path of the PDMP $(\vftilde, \lambdatilde, \Qtilde)$ can be constructed inductively with this approach.

Now we show that the process defined by $(X_t,V_t) = (Y_{r(t)},W_{r(t)})$ for $r(t) = \int_0^t s(Y_{r(u)},W_{r(u)}) \dd u$ can be constructed following the same inductive procedure above, and is therefore a PDMP. For any $n \in \mathbb{N}$, we define $T_n=r^{-1}(\tilde T_n)$, with the convention that $T_0=0$. Assume that the process has been constructed until time $T_n$ for some $n \in \{ 0 , 1 , \dots \}$.

Let us define $\tau_{n+1}=T_{n+1}-T_n$, that is the $(n+1)$-th interarrival time of $(X,V)$. Since $(Y,W)$ follows a deterministic movement in the interval $[\tilde T_n , \tilde T_{n+1})$, so does $(X,V)$ in the interval $[T_n,T_{n+1})$. In particular, for any $t \in [0,\tau_{n+1})$, a simple use of the chain rule yields
\begin{align}
    \frac{\dd \left(X_{T_n+t}, V_{T_n+t}\right)}{\dd t}&= \frac{\dd \left(Y_{r(T_n+t)}, W_{r(T_n+t)}\right)}{\dd t} \\ 
    &=\tilde\Phi\left( Y_{r(T_n+t)}, W_{r(T_n+t)} \right)s\left( Y_{r(T_n+t)}, W_{r(T_n+t)} \right) \\
    &= \tilde \Phi \left(X_{T_n+t},V_{X_{T_n+t}}\right) s\left(X_{T_n+t},V_{X_{T_n+t}}\right) .
\end{align}
Therefore, the deterministic motion of $(X,V)$ is described by the ODE with vector field $\vf = s\vftilde$.
Now, we let us focus on the law of the interarrival time $\tau_{n+1}$. Notice that by a change of variables $k=r(u)$, we have
\begin{align}   
\int_{T_n}^{ T_{n+1}}\lambdatilde\left( X_{u},V_{u} \right) s\left( X_{u},W_{u} \right) \dd u &= \int_{T_n}^{T_{n+1}}\lambdatilde\left( Y_{r(u)},W_{r(u)} \right) s\left( Y_{r(u)},W_{r(u)} \right) \dd u \\
&= \int_{r(T_n)}^{r(T_{n+1})} \lambdatilde\left( Y_k,W_k \right) \dd k \\
& = \int_{\tilde T_n}^{\tilde T_{n+1}} \lambdatilde\left( Y_k,W_k \right) \dd k \\
& =E_{n+1}.
\end{align}
Therefore $\tau_{n+1}$ is the first arrival time of a Poisson process with jump rate $\lambda = s\lambdatilde$.

Finally, at time $T_{n+1}$ the process $(X,V)$ jumps according to the the jump kernel $Q$, which is clearly equal to $\Qtilde$ since $(X_{T_{n+1}},V_{T_{n+1}}) = (Y_{\tilde T_{n+1}},W_{\tilde T_{n+1}})$ and $(X_{T_{n+1}-},V_{T_{n+1}-}) = (Y_{\tilde T_{n+1}-},W_{\tilde T_{n+1}-})$ and hence $(X_{T_{n+1}},V_{T_{n+1}}) = \rmF((X_{T_{n+1}-},V_{T_{n+1}-}),U_{n+1})$.

We have shown that the time changed process $(X,V)$ is defined following the standard construction of a PDMP with characteristics $(\vf,\lambda,Q) = (s\vftilde, s\lambdatilde,\Qtilde)$.

\subsection{Characterisation of time-changed diffusions}\label{proof.diffusion.time.change}
Here we prove  \Cref{prop:timechange_diffusion}.
We denote by $\cLtilde$ the weak generator of the diffusion with characteristics $(\btilde,\sigmatilde)$. \citet{Rogers_Williams_2000} gives that for all $f \in \mathcal{C}^{\infty}_c(\rmE)$, the generator is given by
\begin{equation}
	\cLtilde f(y) = \langle \btilde(y),\nabla f(y)\rangle + \frac12 \sum_{i,j=1}^d \tilde{a}_{ij}(y) \partial_i\partial_j f(y),
\end{equation}
where $\tilde a(x) = \sigmatilde(x) \sigmatilde(x)^T$.  \Cref{thm:generator_timechange} gives that the time changed process has generator 
\begin{equation}
	\cL f(y) = \langle s(y) \btilde(y),\nabla f(y)\rangle + \frac12 \sum_{i,j=1}^d s(y)\tilde{a}_{ij}(y) \partial_i\partial_j f(y),
\end{equation}
for any $f \in \mathcal{C}^{\infty}_c(\rmE)$. Hence the time changed process is a diffusion with characteristics $(s\btilde,\sqrt{s}\sigmatilde)$ by Lemma 1.9, Chapter 5 of \cite{Rogers_Williams_2000}.

\section{Proofs of \Cref{sec:convergence_MP}}
\subsection{Proof of \Cref{thm:ergodicity_markovprocess}}\label{sec:proof_ergod_MP}

    By  \Cref{lem:small_sets} we have that the set $C$ is petite for the process $X$. Since $V \in \mathcal{A}$, we have $\cL V = s \cLtilde V$, hence under \Cref{ass:generic_drift_condition_tt} on $Y_t$ we find 
	\begin{equation}\notag
		\cL \rmV(x) \leq -\speed(x) \rmW(x) \rmV(x) + \gamma \overline{\speed}_{ C } \mathbbm{1}_{ C }(x),
	\end{equation}
	where $\overline{\speed}_{ C }:=\max_{x\in  C } \speed(x)$. The result follows from Theorem 5.2 of \cite{DownMeynTweedie}.

\subsection{Proof of \Cref{thm:uniform_ergodicity_pdmps}}\label{sec:proof_unif_erg_pdmp}
    Denote the states of the PDMP as $z=(x,v) \in\rmE$ and recall that the weak generator of a PDMP $(\vftilde,\lambdatilde,\Qtilde)$ is of the form \citep{Davis1993}
    \begin{equation}
        \cLtilde f(z) = \langle \vftilde(z),\nabla f(z)\rangle + \lambdatilde(z)(\Qtilde f(z)-f(z)).
    \end{equation}
	Let $\bar{\rmV}(z)=\rmV(z)/(1+\rmV(z))$, where $\rmV$ is a Lyapunov function satisfying \eqref{eq:generic_drift} for the PDMP $(\vftilde,\lambdatilde,\Qtilde)$. We note that $\bar{\rmV}$ is bounded above by $1$, since $\rmV\geq 1$. We shall show that $\bar \rmV$ also satisfies a drift condition of the form \eqref{eq:generic_drift}.
	Applying the generator of $Y$ to $\bar{\rmV}$ we find
	\begin{align}
		\cLtilde \bar{\rmV}(z) & = \frac{1}{(1+\rmV(z))^2} \langle \vftilde (z),\nabla \rmV(z)\rangle +  \lambdatilde(z) \int (\bar{\rmV}(y)-\bar{\rmV}(z)) \Qtilde (z,\dd y).
	\end{align}
	For the jump part we obtain
	\begin{align}
		\lambdatilde(z) \int (\bar{\rmV}(y)-\bar{\rmV}(z)) \Qtilde(z,\dd y)  &= \lambdatilde(z) \int \frac{\rmV(y)-\rmV(z)}{(1+\rmV(y)) (1+\rmV(z))} \Qtilde(z,\dd y) 
	\end{align}
	Considering both the case $\rmV(y)\geq \rmV(z)$ and $\rmV(z)\geq \rmV(y)$ we obtain 
	\begin{align}
		\frac{\rmV(y)-\rmV(z)}{(1+\rmV(y)) (1+\rmV(z))} \leq  \frac{\rmV(y)-\rmV(z)}{(1+\rmV(z))^2}.
	\end{align}
	Therefore applying our assumptions we find 
	\begin{align}
		\cLtilde \bar{\rmV}(z) & \leq  \frac{1}{(1 + \rmV(z))^2} \Big( \langle \vftilde(z),\nabla \rmV(z)\rangle + \lambdatilde(z) \int (\rmV(y)-\rmV(z)) \Qtilde(z,\dd y)  \Big)\\
		& = \frac{1}{(1+\rmV(z))^2}\,\, \cLtilde \rmV(z)\\
		& \leq \frac{1}{(1+\rmV(z))^2}  \Big( - \rmW(z) \rmV(z) + \gamma \1_{ C }(z) \Big)\\
		& \leq - \frac{\rmW(z) }{1+\rmV(z)} \, \bar{\rmV}(z) + \gamma \1_{ C }(z) .
	\end{align}

Since outside of $ C $ we have $s \geq \beta \rmV/ \rmW $, and $s$ is bounded above by $\bar{s}$ on $C$, the result follows applying \Cref{thm:ergodicity_markovprocess}.

\section{Proof of \cref{thm:FCLT}}\label{sec:proof_FCLT}
\Cref{eq:FCLT_convergence} follows noticing that the conditions of \citet[Theorem 4.3]{glynn1996liapounov} are verified under our assumptions. By \citet[Theorem 4.3]{glynn1996liapounov} we find that $\gamma_g^2 = 2 \int_\rmE \hat g(z) \overline g(z)  \mu(\dd z),$
        where $\overline g(z) =g(z)-\mu(g)$, while $\hat g$ is the solution to the Poisson equation $\overline g(z) = - \cL \hat g(z)$. 
        
        We now prove \eqref{eq:asymptotic_variance}. Consider the function $f = \nicefrac{\overline g}{s},$ which satisfies $\mutilde(f)=0.$ 
        First, notice that we can apply \citet[Theorem 4.3]{glynn1996liapounov} to the base process $\cLtilde$ under the additional condition $\rmW(z)\rmV(z) \geq 1$ for all $z\in\rmE$. 
        Then we find that the asymptotic variance of $\cLtilde$ for the function $f$ is \[ \gammatilde^2_{f} = 2\int \hat f(z) f(z) \mutilde(\dd z), \quad \text{where } f = -\cLtilde \hat f.\]
        Since $\hat g \in \mathcal{A}$, we get $\overline g= -\cL \hat g = - s \cLtilde \hat g$, therefore $f=-\cLtilde \hat g$. From the uniqueness of the solution of the Poisson equation \citep{glynn1996liapounov} we get that $\hat g = \hat f$.  
        We find that 
        \begin{align}
            \gamma_g^2 &= 2 \int_\rmE \hat g(z) (- s(z) \cLtilde \hat g(z))  \mu(\dd z) \\
            & = 2 \mu(s)  \int_\rmE \hat f(z) (-  \cLtilde \hat f(z))  \frac{s(z)}{\mu(s)}\mu(\dd z)\\
            & = 2 \mu(s)  \int_\rmE \hat f(z) (-  \cLtilde \hat f(z)) \mutilde(\dd z)\\
            & = \mu(s) \,\gammatilde^2_{f}.
        \end{align}
        The final step is to notice that  for any $\alpha>0$ it holds that $\gammatilde^2_{\alpha g} = \alpha^2 \gammatilde^2_g$, and therefore 
        \[ \mu(s) \, \gammatilde^2_{f} = \gammatilde^2_{f \sqrt{\mu(s)} }\,.\]

\section{Proofs for the time-changed ZZP}\label{appendix:ZZproofs}

\subsection{Eyring-Kramers formula for the time changed ZZP}\label{sec:eyring-kramers}

We now obtain an Eyring-Kramers formula for the one dimensional, time changed ZZP in the context of Section \ref{sec:bimodal_example}. This type of result was first obtained in\cite[Theorem 1.1]{Monmarche2016} for the standard ZZP.
We consider a target of the form $\pot_\varepsilon = \pot/\varepsilon$, where $\pot$ is a double well potential, and we characterise the behaviour of the process as $\varepsilon\to 0$. Let $\tau:=\inf\{t>0: X_t=x_1 \},$ i.e. the time elapsed before the process finally moves from its initial condition $(X_0,V_0)=(x_0,-1)$ to $x_1$. In the next proposition we state our result for speed functions of the form $s_\varepsilon(x) = s(x)^{1/\varepsilon}$, where $s$ is a fixed speed function, satisfying the properties discussed in \Cref{sec:bimodal_example}.
\begin{proposition}\label{thm:eyring_kramers_suzz}
	Consider a smooth one-dimensional double well potential $\pot>0$ with $\pot''(x_0)>0$ and let $s:\mathbb{R} \to [1,\infty)$ be a smooth function. Define $\pottilde(x)=\pot(x)-\ln s(x)$ and assume $\pottilde$ is a double well potential with same maxima and minima of $\pot$.
	Then the time changed ZZP with target $\mu_\varepsilon\propto \exp(-\pot/\varepsilon)$ and speed function $s_\varepsilon(x) = s(x)^{1/\varepsilon}$ satisfies
	\begin{equation}\notag
		\mathbb{E}_{(x_0,-1)}[\tau] \leq \sqrt{\frac{8\mu\varepsilon}{\pottilde''(x_0)}} \exp\left( \frac{\pottilde(x_1)-\pottilde(x_0)}{\varepsilon} \right) (1+o\left(\varepsilon\right) ),
	\end{equation}
    where $o\left(\varepsilon\right)$ indicates a function that tends to zero as $\varepsilon$ tends to zero.
\end{proposition}
\begin{proof}[Proof of Theorem \ref{thm:eyring_kramers_suzz}]
	Let $\tilde\tau = r(\tau)$ be the first time at which the ZZP with stationary distribution $\propto s_{\varepsilon}\mu_\varepsilon=\exp\left\{ -\left( U_{\varepsilon} -\ln s_{\varepsilon} \right) \right\}$ reaches $x_1$. Because $s\geq 1$, we find $\mathbb{E}_{(x_0,-1)}[\tau] \leq  \mathbb{E}_{(x_0,-1)}[\tilde\tau].$
	Under our assumptions $\pot-\ln s$ is itself a double well potential and hence we can apply \cite[Theorem 1.1]{Monmarche2016} to obtain the result.
\end{proof}
\begin{remark}
	Let us compare the result of \Cref{thm:eyring_kramers_suzz} between the case of the standard ZZP (i.e. choosing $s(x) = 1$)
    and when $s(x)=\exp(a\pot(x))$ for $a\in(0,1)$. In the former case we obtain from \cite[Theorem 1.1]{Monmarche2016} 
    that
    \begin{equation}\notag
\mathbb{E}_{(x_0,-1)}[\tau] \leq \sqrt{\frac{8\mu\varepsilon}{\pot''(x_0)}} \exp\left( \frac{\pot(x_1)-\pot(x_0)}{\varepsilon} \right) (1+o\left(\varepsilon\right) )  ,
    \end{equation} 
    while in the latter case we obtain
	\begin{equation}\notag
		\mathbb{E}[\tau] \leq  \sqrt{\frac{8\mu\varepsilon}{(1-a)\pot''(x_0)}} \exp\left( (1-a)\frac{\pot(x_1)-\pot(x_0)}{\varepsilon} \right) (1+o\left(\varepsilon\right) ).
	\end{equation}
	As $\varepsilon$ tends to $0$, we see that the upper bound on the hitting time of $x_1$ of SUZZ increases with a lower rate as that of ZZS, indicating an improved performance of the time-transformed process when targeting a multi-modal distribution with very high density at the modes.
\end{remark}

\subsection{Proof of \Cref{thm:ergodicity_ZZP}}\label{sec:proof_ergo_ZZP}
We first prove the statement on geometric ergodicity. This can be obtained applying \Cref{thm:ergodicity_markovprocess} and hence we verify the required assumptions. We have assumed that \Cref{ass:s.integrability} holds in the statement of the Proposition.
Irreducibility of the ZZP was shown under our conditions by \cite{Bierkensergodicity}. Aperiodicity of the time-changed ZZP (i.e. \Cref{ass:X_is_aperiodic}) was shown in Proposition F.2 of \cite{Vasdekis_speedup}. \Cref{ass:lln_base} is proven to hold under our assumptions in \cite{Bierkensergodicity}.
We now prove \Cref{ass:generic_drift_condition_tt}(1). 
Lemma 4.1 of \cite{bertazzi2020adaptive} establishes that, when $\gamma_i(x) $ are lower and upper bounded, any set of the form $ C = C _x  \times  C _v$ where $ C _x$ is a compact set and $ C _v \subset \{\pm 1\}^d$ is $(t,b(t),\nu)$-small for the ZZP for any $t\geq t_0$, where $t_0$ is a large enough time and $b(t)$ is continuous in $t$. Consider a bounded set $  C  \supset C$. Inspecting the proof of \cite[Lemma 4.1]{bertazzi2020adaptive} one can see that \Cref{ass:generic_drift_condition_tt}(1) holds as long as $\gamma_i$ is bounded on the set $  C $, since this gives sufficient control over the switching rates. 
Next, we focus on \Cref{ass:generic_drift_condition_tt}(2). Lemma 11 in \cite{Bierkensergodicity} gives that the following function satisfies the drift condition \eqref{eq:generic_drift} for $\rmW(z)=\eta$ for some $\eta>0$:
\begin{equation}\label{Lyapunov.for.ZZ}
	\rmV(x,v) = \exp\left( \alpha \pottilde(x) + \sum_{i=1}^d \phi(v_i (1-\beta)\partial_i\pottilde(x)) \right),
\end{equation}
where $\delta,\alpha >0 $ are such that $0 < \delta \overline{\gamma} < \alpha < 1$ and where $\overline{\gamma}:=\max_{i,x} \gamma_i(x) / s(x)$, and finally $\phi(u) = \tfrac{1}{2} \textnormal{sign}(u) \ln{(1+\delta |u|)}$. Therefore, \Cref{thm:ergodicity_markovprocess} guarantees geometric ergodicity since $s$ is lower bounded.

Let us now prove the statement on uniform ergodicity. Our goal is to apply \Cref{thm:uniform_ergodicity_pdmps}. 
For $s(x) = \exp(\beta\pot(x))$ for $\beta\in(0,1)$ we have that the base process targets  the density $\mutilde(\dd x , \dd v) \propto \exp\left\{ (1-\beta) \pot(x) \right\}$. Under \Cref{ass:ZZ_growth_proposal}, by Lemma 11 of \cite{Bierkensergodicity} we have that the standard ZZP with target $\mutilde$ and excess switching rates $\gammatilde_i$ satisfies the drift condition \eqref{eq:generic_drift} with $\rmW(z)= \eta$ for some $\eta>0$, and Lyapunov function 
\begin{equation}\notag
	\rmV_\beta(x,v) = \exp\left( \alpha(1-\beta)\pot(x) + \sum_{i=1}^d \phi(v_i (1-\beta)\partial_i\pot(x)) \right).
\end{equation}
In particular, $\rmV$ is a Lyapunov function for arbitrarily small values of the constant $\alpha$. Therefore, it is now sufficient to choose $\alpha\in(0,1)$ such that outside of $ C $
\begin{equation}\notag
	s(x)= \exp(\beta\pot(x)) \geq b \, \rmV_\beta(x,v),
\end{equation}
for some $b>0$.
This is possible since, by \Cref{ass:ZZ_growth_proposal}, $\pot$ is the leading term in the exponent of $\rmV$ and the set $ C $ can accordingly be chosen large enough.
Therefore, the time-changed ZZP is uniformly ergodic for our choice of speed function.

\section{Proof of \Cref{prop:ergodicity_jumpproc}}\label{sec:proof_ergo_jumpproc}
We apply \Cref{thm:ergodicity_markovprocess} and hence we need to verify the corresponding conditions.
    
First of all, \Cref{ass:s.integrability} is assumed to hold and \Cref{ass:lln_base} holds due to \Cref{ass:geoerg_discretetime}(1). This also implies that $Y$ is $\psi$-irreducible and by \Cref{prop:irreducibility}, the same holds for $X$. 

    Now we prove that the set $ C $ in \Cref{ass:geoerg_discretetime}(2) is petite for $X.$ By \Cref{lem:small_sets}, it suffices to show that \Cref{ass:generic_drift_condition_tt}(2) holds for the process $Y$.
    For any measurable set $A$ and $y \in C$, we write
	\begin{align}
		\mathbb{P}_y ( Y_{t} \in A, Y_u\in   D  \textnormal{ for all } u\in [0,t]  ) \geq \mathbb{P}_y ( \Ybar_{n_*} \in A, \Ybar_n \in   D  \textnormal{ for all } n=0,1,\dots,n_*,\, t\in[\Ttilde_{n_*},\Ttilde_{n_*+1}]  ).
	\end{align}
    Notice the event $t\in[\Ttilde_{n_*},\Ttilde_{n_*+1}]$ is independent of $\Ybar_n$ for all $n$ as it only depends on i.i.d. drawn exponential random variables with parameter $1$.
	Therefore, by \Cref{ass:geoerg_discretetime}(2), for any $t_0,\varepsilon>0$ we find
	\begin{align}
		\int_{t_0}^{t_0+\varepsilon} \mathbb{P}_y ( Y_{t} \in A, Y_u\in D \textnormal{ for all } u\in [0,t]  )\dd t\geq c\nu(A) \int_{t_0}^{t_0+\varepsilon}  \mathbb{P}(t\in[\Ttilde_{n_*},\Ttilde_{n_*+1}]) \dd t = c'(t_0,\varepsilon) \,\nu(A).
	\end{align}
    Hence \Cref{ass:generic_drift_condition_tt}(2) is satisfied and $ C $ is petite for $X.$
    
    Aperiodicity of $X$ (i.e. \Cref{ass:X_is_aperiodic}) follows observing that for any $x\in C $ it holds that
    \begin{align}
        \PP_x(X_t\in  C ) &\geq \PP_x(X_t\in  C , \,t<T_1)\\
        & = \PP_x(t<T_1)\\
        & \geq \exp\left( -t \sup_{z\in C } s(z) \right) 
    \end{align}
    Since under our assumptions $\sup_{z\in C } s(z) \leq \bar s$ we conclude that $\PP_x(X_t\in  C ) >0$ for all $x\in C $ and all $t\geq 0.$
    
    Finally, we show the drift condition for $Y$ required by \Cref{thm:ergodicity_markovprocess} holds.
    \Cref{ass:geoerg_discretetime}(3) gives that the process $Y$ satisfies
    \begin{align}
        \cLtilde \rmV(z) &= \Qtilde \rmV(z) - \rmV(z)\\
        & \leq - (1 - \rmW(z)) \rmV(z) + \eta\1_{ C }, 
    \end{align}
    that is a drift condition of the form \eqref{eq:generic_drift}.
    
    Therefore, we can apply \Cref{thm:ergodicity_markovprocess} to obtain that $X$ is geometrically ergodic if there exists $\beta >0$ such that $s(x)\geq \nicefrac{\beta}{1-\rmW(x)}$ for all $x\notin C$.
	Moreover, since $Y$ and $X$ are PDMPs, Theorem \ref{thm:uniform_ergodicity_pdmps} gives uniform ergodicity of $X$ when there exists $\beta > 0$ such that $s(x) \geq \nicefrac{\beta\rmV(x)}{1-\rmW(x)}$ for all $x\notin C .$
	
	\Cref{ass:geoerg_discretetime}(1) guarantees that the only stationary distribution is $\mu$, as shown in \Cref{thm:lln_timechange}.

\section{Additional numerical simulations}
\Cref{fig:gaussianmixture_ctstime} shows further numerical simulations for the mixture of normal distribution experiment of \Cref{ex:gaussianmixture}. In particular, the standard ZZP with tempered target can be simulated by Poisson thinning (see \cite{Sutton_continuously_tempered}). A realisation is shown in \Cref{fig:zzp_mixture}. As mentioned in \Cref{sec:estimate_expectations}, we can then use a discretisation of the obtained path to construct a jump process with rate $s$. The result of this procedure is shown in \Cref{fig:tc-zzp_mixture}.

\begin{figure}
    \begin{subfigure}[t]{0.4\textwidth}
		\includegraphics[width=\textwidth]{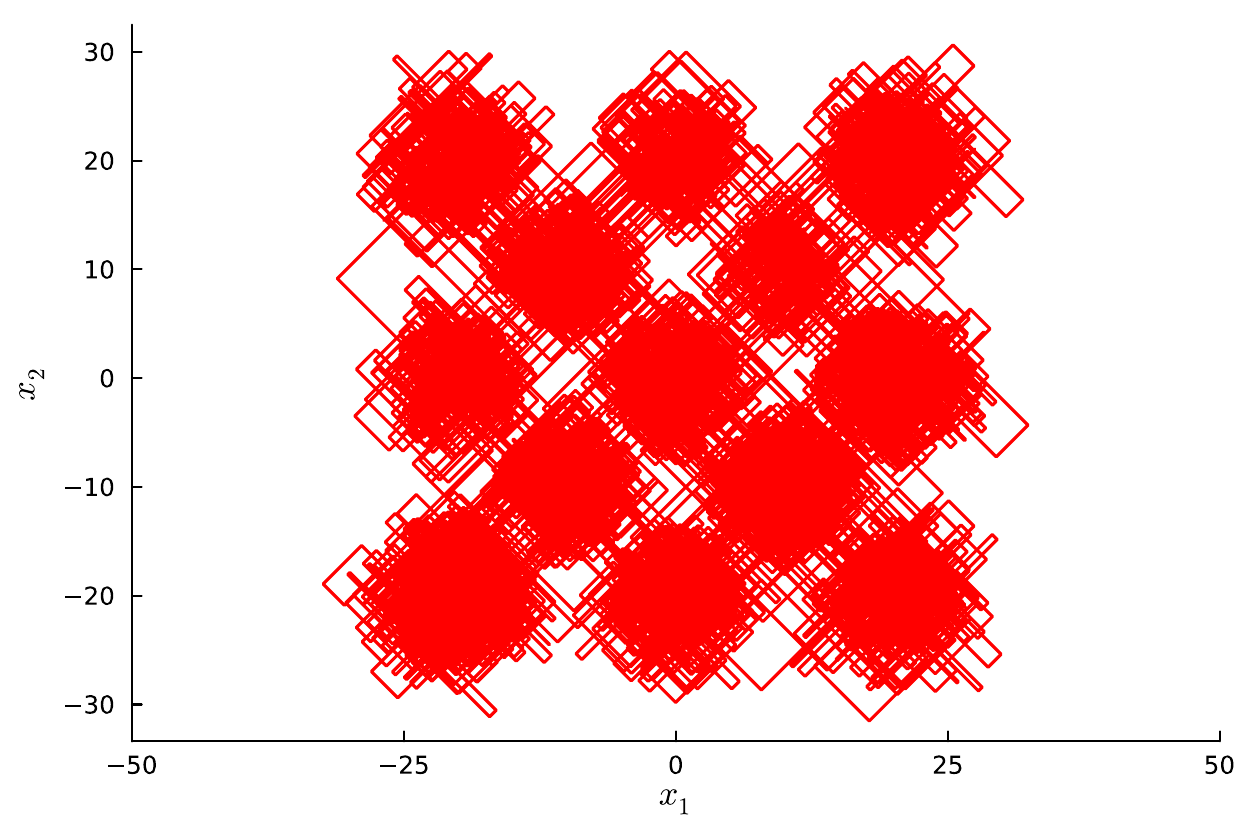}
		\caption{Trace plot of a realisation of a ZZP with target $\mutilde(x)=\mu(x)^{0.1}$.}
		\label{fig:zzp_mixture}
	\end{subfigure}
	\hspace{10pt}
	\begin{subfigure}[t]{0.4\textwidth}
		\includegraphics[width=\textwidth]{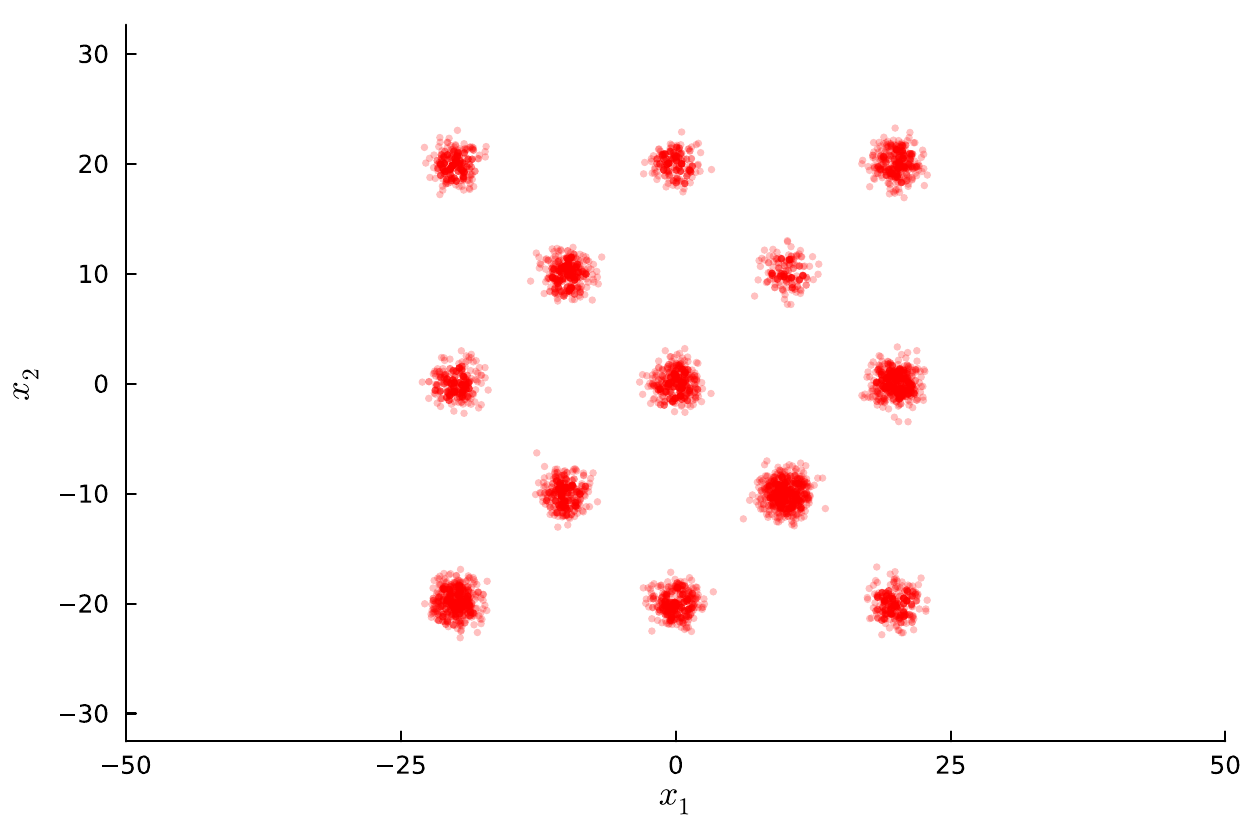}
		\caption{Discretisation of the time-changed, jump process.}
		\label{fig:tc-zzp_mixture}
	\end{subfigure}
	\caption{Further numerical simulations for a mixture of normal distributions as in \Cref{ex:gaussianmixture}. The process in the left plot has time horizon $3\times 10^4$ and initial condition at the origin. The jump process on the right plot is obtained first discretising the continuous path of the ZZP with step size $2,$ then discretising the path of the obtained jump process with step size $5e-2.$}
	\label{fig:gaussianmixture_ctstime}
\end{figure}

\end{document}